\documentclass{elsarticle}

\usepackage{amsmath,amsthm,amssymb,   
            filecontents,		  	  
            changepage,               
            calc,                     
            tikz,                     
            ifthen,
            xspace,
            hyperref}                 
            
\usetikzlibrary{calc,                 
                matrix}               

\newtheorem{thm}{Theorem}
\newtheorem{defn}{Definition}
\newtheorem{lem}[defn]{Lemma}
\newtheorem{prop}[defn]{Proposition}

\newtheorem{prob}{Problem}
\newtheorem{claim}{Claim}
\newtheorem{subclaim}{Subclaim}

\newcommand{\N}{\mathbb{N}}
\newcommand{\pr}{\mathbb{P}}

\newcommand{\C}{\mathcal{C}}
\newcommand{\F}{\mathcal{F}}
\newcommand{\U}{\mathcal{U}}
\newcommand{\lAP}{\le_{\textnormal{AP}}}
\newcommand{\eAP}{\equiv_{\textnormal{AP}}}
\def\numP{\#\textnormal{P}\xspace}
\newcommand{\nae}{\textnormal{NAE}\xspace}     
\newcommand{\tnae}{\textnormal{3-NAE}\xspace}  

\newcommand{\DS}{\textnormal{DS}}        
\newcommand{\IMS}{\textnormal{IMS}}      
\newcommand{\IS}{\textnormal{IS}}        
      
\newcommand{\MES}{\textnormal{MES}}      
\newcommand{\MIS}{\textnormal{MIS}}      
\newcommand{\MS}{\textnormal{MS}}        
\newcommand{\SAT}{\textnormal{SAT}}      
\newcommand{\SU}[1]{|\U(#1)|}            
\newcommand{\TMES}{\textnormal{LMES}}    
\newcommand{\UR}[1]{|\U_{\F}^{-1}(#1)|}  
\newcommand{\VC}{\textnormal{VC}}        

\newcommand{\bds}{\textnormal{\#BiDomSets}\xspace}
\newcommand{\bims}{\textnormal{\#Bi\-Inclusion\-Minimal\-Seps}\xspace}
\newcommand{\bmes}{\textnormal{\#Bi\-Minimal\-Edge\-Seps}\xspace}
\newcommand{\bms}{\textnormal{\#Bi\-Minimal\-Seps}\xspace}
\newcommand{\bmstwo}{\textnormal{\#BiConn\-Minimal\-Seps}\xspace}
\newcommand{\bis}{\textnormal{\#BIS}\xspace} 
\newcommand{\is}{\textnormal{\#IS}\xspace}
\newcommand{\largemes}{\textnormal{\#Large\-Minimal\-Edge\-Seps}\xspace}
\newcommand{\largestmes}{\textnormal{\#$(s,t)$-Large\-Minimal\-Edge\-Seps}\xspace}

\newcommand{\maxlbis}{\textnormal{\#MaximalBIS}\xspace}

\newcommand{\rhp}{\textnormal{\#RH$\Pi_1$}\xspace}
\newcommand{\sat}{\textnormal{\#SAT}\xspace}

\newcommand{\setunion}{\textnormal{\#SetUnions}\xspace} 

\newcommand{\stbmes}{\textnormal{\#$(s,t)$-Bi\-Minimal\-Edge\-Seps}\xspace}
\newcommand{\stbms}{\textnormal{\#$(s,t)$-Bi\-Minimal\-Seps}\xspace}
\newcommand{\stbmstwo}{\textnormal{\#$(s,t)$-Bi\-Conn\-Minimal\-Seps}\xspace}
\newcommand{\tnaesat}{\textnormal{\#Monotone\-Promise-3-NAE-SAT}\xspace}
\newcommand{\unionrep}{\textnormal{\#UnionReps}\xspace}
\newcommand{\vc}{\textnormal{\#VertexCovers}\xspace} 

\def\myprob#1#2#3{
 #1.\\
\textbf{Input:} #2.\\
\textbf{Output:}  #3.}

\tikzstyle{vertex} = [circle, inner sep = 0.85mm, fill=black]
\tikzstyle{smallvertex} = [circle, inner sep = 0.4mm, fill=black]
\tikzstyle{ellipses} = [circle, inner sep = 0.15mm, fill=black]
\tikzstyle{occupied} = [draw, color = red, inner sep = 1.9mm]
\tikzstyle{smalloccupied} = [draw, color = red, inner sep = 1mm]

\definecolor{darkred} {rgb}{0.95, 0,   0}
\definecolor{green}   {rgb}{0,    0.9, 0}
\definecolor{blue}    {rgb}{0,    0,   0.9}
\definecolor{grey}    {rgb}{0.9,  0.9, 0.9}
\definecolor{medgrey} {rgb}{0.8,  0.8, 0.8}
\definecolor{darkgrey}{rgb}{0.5,  0.5, 0.5}

\hypersetup{linkbordercolor=blue}

\begin{document}

\title {Approximately Counting Locally-Optimal Structures\tnoteref{t1}}
\tnotetext[t1]{
	The research leading to these results has received funding from
	the European Research Council under the European Union's Seventh
	Framework Programme (FP7/2007--2013) ERC grant agreement
	no.\ 334828. The paper reflects only the authors' views and not
	the views of the ERC or the European Commission. The European
	Union is not liable for any use that may be made of the
	information contained therein. The research leading to these results has also received funding from the National Science Foundation, grant number IIS-1219278. A preliminary version~\cite{ICALP} of this
	paper appears in the proceedings of ICALP 2015.}

\author[oxf]{Leslie~Ann~Goldberg}
\ead{leslie.goldberg@cs.ox.ac.uk}

\author[rg]{Rob~Gysel}
\ead{rsgysel@ucdavis.edu}

\author[oxf]{John~Lapinskas\fnref{fn1}}
\ead{john.lapinskas@cs.ox.ac.uk}

\address[oxf]{Department of Computer Science, University of Oxford, Parks Road, OX1 3QD, UK.}
\address[rg]{Department of Computer Science, University of California, 2063 Kemper Hall, One Shields Avenue, Davis, CA 95616-8562, US.}

\fntext[fn1]{Corresponding author.}

\begin{abstract}
		In general, constructing a locally-optimal structure is a little harder than constructing an arbitrary structure, but significantly easier than constructing a globally-optimal structure. A similar situation arises in listing. In counting, most problems are \#P-complete, but in approximate counting we observe an interesting reversal of the pattern. Assuming that \#BIS is not equivalent to \#SAT under AP-reductions, we show that counting maximal independent sets in bipartite graphs is harder than counting maximum independent sets. Motivated by this, we show that various counting problems involving minimal separators are \#SAT-hard to approximate. These problems have applications for constructing triangulations and phylogenetic trees.
\end{abstract}

\maketitle{}

\section{Introduction}\label{sec:intro}

A \emph{locally-optimal} structure is  a combinatorial structure that cannot be improved
by certain (greedy) local moves, even though it may not be globally optimal.
An example is a maximal independent set in a graph.
It is trivial to construct an independent set in a graph (for example, the singleton set containing any vertex
is an independent set).
It is easy to construct a maximal independent set (the greedy algorithm
can do this). However, it is NP-hard to construct a globally-optimal  independent set, 
which in this case means a
maximum independent set.
In the setting in which we work, this situation is typical.  Constructing a locally-optimal structure is somewhat more difficult
than constructing an arbitrary structure, and   constructing a globally-optimal structure
is more difficult than constructing a locally-optimal structure.
For example, in bipartite graphs, it is trivial to construct an independent set,
easy to (greedily) construct a maximal independent set, and more difficult to construct
a maximum independent set (even though this can be done in polynomial time).
This general phenomenon has been well-studied. In 1987, Johnson, Papadimitriou and
Yannakakis~\cite{JYPPLS} defined the complexity class PLS (for ``polynomial-time local search'')
that captures local optimisation problems where one iteration
of the local search algorithm takes polynomial time. 
As the authors point out, practically all empirical evidence leads to the conclusion
that finding locally-optimal solutions is much easier than 
solving NP-hard problems,
and this is supported by complexity-theoretic evidence, since  a problem in PLS cannot be NP-hard unless NP=co-NP.
An example that illustrates this point is the graph partitioning problem.
For this problem it is trivial to find a valid partition, and it is NP-hard to
find a globally-optimal (minimum weight) partition but
Sch\"{a}ffer and Yannakakis~\cite{sy-pls} showed that 
finding a locally-optimal solution (with respect to a particular swapping-dynamics)
is PLS-complete, so is presumably of intermediate complexity.

For listing combinatorial structures, a similar pattern emerges. 
By self-reducibility, there is a nearly-trivial polynomial-space polynomial-delay algorithm for listing the independent
sets of a graph~\cite{leslielisting}. 
A polynomial-space polynomial-delay algorithm for listing the \emph{maximal} independent
sets exists, due to Tsukiyama et al.~\cite{tsukiyama}, but it is more complicated.
On the other hand, there is no polynomial-space polynomial-delay algorithm for listing
the \emph{maximum} independent sets unless P=NP.
There is a polynomial-space polynomial-delay algorithm for
listing the maximum independent sets of a bipartite graph~\cite{kmnf}, but this is substantially more 
complicated than any of the previous algorithms. 

When we move from constructing and listing to counting, these differences
become obscured because nearly everything is $\numP$-complete.
For example, counting independent sets, maximal independent sets, and maximum independent
sets of a graph are all $\numP$-complete problems, even if the graph is bipartite~\cite{vadhan}. 
Furthermore, even \emph{approximately} counting independent sets, maximal independent sets, and
maximum independent sets of a graph
are all   $\numP$-complete with respect to approximation-preserving reductions~\cite{dggj-approx}.

The purpose of this paper is to highlight
an interesting situation that arises in approximate counting
where, contrary to the situations that we have just discussed,
approximately counting locally-optimal structures is apparently more difficult than counting
globally-optimal  structures.

In order to explain the result, we first briefly summarise what is known about the
complexity of approximate counting
within $\numP$. This will be explained in more detail in Section~\ref{sec:notation}.
There are three relevant complexity classes ---
the class containing problems which admit a fully-polynomial randomised approximation scheme (FPRAS),
the class \rhp, and $\numP$~itself. 
 Dyer et al.~\cite{dggj-approx} showed
that $\bis$, the problem of counting independent sets in a bipartite graph,
is complete 
for~\rhp\
with respect to approximation-preserving (AP) reductions
and that $\is$, the problem of counting independent sets in a (general) graph
is  $\numP$-complete with respect to AP-reductions.
It is generally believed that the   \rhp-complete problems
are not FPRASable, but that they are 
of intermediate complexity, and are not as difficult to approximate as the problems
which are  $\numP$-complete with respect to AP-reductions.
Many problems have subsequently been shown to be  \rhp-complete and  \numP-complete
with respect to AP-reductions. More examples will be given in Section~\ref{sec:notation}.

We can now describe the interesting situation which emerges with respect to
independent sets in bipartite graphs.
Dyer et al.~\cite{dggj-approx} showed that approximately counting independent 
sets and approximately counting \emph{maximum} independent sets 
are both  \rhp-complete with respect to AP-reductions.
Thus, the pattern outlined above would suggest
that approximately counting \emph{maximal} independent sets in
bipartite graphs   ought to also be \rhp-complete.
However, we show 
(Theorem~\ref{thm:maxl-bis-hard}, below)
that approximately counting \emph{maximal} independent
sets in bipartite graphs is actually   \numP-complete with respect to AP-reductions.
Thus, either
\rhp\ and \numP\ are equivalent in approximation complexity (contrary to the picture that has been 
emerging in earlier papers), or
this is a scenario where approximately counting locally-optimal structures 
is actually more difficult than approximately counting globally-optimal ones.

Motivated by the difficulty of approximately counting maximal independent sets
in bipartite graphs, we  also study the problem of approximately counting
other  locally-optimal structures that  
arise in  algorithmic applications.
First, the problem of counting
the \emph{minimal separators} of a graph arises in diverse applications
from triangulation theory to phylogeny construction in computational biology.
A minimal separator is a particular type of vertex separator. Definitions are given in Section~\ref{sec:OurResults}.
Algorithmic applications arise because
fixed-parameter-tractable algorithms are known whose running time
is   polynomial in the number of minimal separators of a graph.
These algorithms 
were originally developed by Bouchitt{\'e} and Todinca \cite{BT01,BT02} (and improved in \cite{FKTV08}) to exactly solve the so-called \emph{treewidth} and \emph{minimum-fill} problems.
The former problem, finding the exact treewidth of a graph, is widely studied due to its applicability to a number of other NP-complete 
problems~\cite{BK08}.
The technique has recently been generalized~\cite{FY14} to cover 
problems including \emph{treecost}~\cite{BF05} and \emph{treelength}~\cite{L10}.
The algorithm can also be used to find a 
minimum-width
\emph{tree-decomposition} of a graph, a key data structure that is used to solve a variety of NP-complete problems in polynomial time when the width of the tree-decomposition is fixed \cite{BK08}.
In recent years, much research has been dedicated to exact-exponential algorithms for treewidth \cite{BFKKT12}, the fastest of which \cite{FV10} 
has running time closely connected to the number of minimal separators in the graph. Indeed, there exist polynomials $p_L$ and $p_U$ such that if the graph has $n$ vertices and $M$ minimal separators, then the running time is at least $p_L(n)M$ and at most $p_U(n)M^2$.

Bouchitt{\'e} and Todinca's approach has also recently been applied to solve the \emph{perfect phylogeny problem} and two of its variants \cite{G14}.
In this problem, the input is a set of phylogenetic characters, each of which may be viewed as a partition of a subset of \emph{species}.
The goal is to find a phylogenetic tree 
such that every character is \emph{convex} on that tree --- that is, the parts of each partition form connected subtrees that do not overlap.
Such a tree is called a \emph{perfect phylogeny}.

In all of these applications, it would be useful to count the minimal vertex separators of
a graph, since this would give an a priori bound on the running time of the algorithms.
Thus, we consider the difficulty of this problem, whose complexity was previously unresolved, even in terms
of exact computation.
Theorem~\ref{thm:ms-vertex-hard} shows that the problem of counting minimal separators
is $\numP$-complete, both with respect to Turing reductions (for exact computation) and
with respect to AP-reductions. Thus, this problem is as difficult to approximate as any problem in~$\numP$.

Motivated by applications to treewidth~\cite{FKTV08} and phylogeny~\cite{RobPreprint, G14},
we also consider various heuristic approximations to the minimal separator problem.
The number of inclusion-minimal separators is a natural choice for a lower bound on the number
of minimal separators. Conversely, the number of $(s,t)$-minimal separators, taken over all
vertices~$s$ and~$t$, is a natural choice for an upper bound on the number of minimal separators.
Theorem~\ref{thm:ms-vertex-hard} shows that both of these bounds are difficult to
compute, either exactly or approximately.
Finally, the number and structure of $2$-component minimal separators
is important in computational biology. $2$-component minimal separators
arise naturally in the problem of determining whether a subset of
``quartet phylogenies'' can be assembled uniquely~\cite{RobPreprint}.
Thus, we study the problem of counting such minimal separators.
Theorem~\ref{thm:ms-vertex-hard} shows that they are 
complete for $\numP$ with respect to exact and approximate computation.

Our new results about counting minimal vertex separators are obtained by first considering the problem
of counting minimal edge separators. These locally-optimal structures are also known as \emph{bonds} or \emph{minimal cuts}, and are well-studied in other contexts --- see e.g. Diestel~\cite{diestel}.
Theorem~\ref{thm:ms-edge-hard} gives the first hardness result for counting these
structures, either exactly or approximately.

In addition to studying maximal independent
sets and minimal vertex and edge separators,
we study  two other locally-optimal structures related
to maximal independent sets in bipartite graphs. 
A maximal independent set is precisely an independent set in a graph which is also a
\emph{dominating set}.  Theorem~\ref{thm:bds-hard} shows that counting dominating sets
in bipartite graphs is $\numP$-hard with respect to AP-reductions.
It is already known to be $\numP$-hard to compute exactly~\cite{kou}.
Finally, in Theorems~\ref{thm:su-hard} and
\ref{thm:ur-hard} we show that maximal independent sets in bipartite graphs can be represented as
unions of sets, so a set union problem \setunion\ is also
$\numP$-hard with respect to AP-reductions, and so is its inverse \unionrep.

 \subsection{Detailed Results}
 \label{sec:OurResults}

We now give  formal definitions of the problems that we study,
and state our results  precisely. Note that all problems are  indexed for reference   at the end of the paper.
Our first result is that counting maximal independent sets in a bipartite
graph is $\numP$-complete with respect to AP-reductions (even though counting
 maximum independent sets 
in bipartite graphs  is only \rhp-complete with respect to these reductions).
For readers that are unfamiliar with AP-reductions, details
 are given in Section~\ref{sec:notation}.

\begin{defn}\label{defn:maxl-is}
Let $G$ be a graph. We say that an independent set $X \subseteq V(G)$ of~$G$ is \emph{maximal} if 
no proper superset of $X$ is an independent set of~$G$.
\end{defn}

\begin{prob}\label{prob:maxlbis}\myprob{\maxlbis}{A bipartite graph $G$}{The number of maximal
independent sets of~$G$}
\end{prob}
 
The following theorem is proved in Section~\ref{sec:maxl-bis}.

\begin{thm}\label{thm:maxl-bis-hard}
$\maxlbis \eAP \sat$.
\end{thm}

Next we state our results relating to counting minimal separators.
In the following definitions, $G=(V,E)$ is a graph, $s$ and $t$ are distinct vertices of~$G$,
and $X\subseteq V$ is a set of vertices.

\begin{defn}   $X$ is an \emph{$(s,t)$-separator} of $G$ if $s$ and $t$ lie in different components of $G-X$. 
If, in addition, no proper subset of $X$ is an $(s,t)$-separator of~$G$, then
we say that $X$ is a   \emph{minimal $(s,t)$-separator}  of~$G$.
\end{defn}

\begin{defn} $X$ is a \emph{minimal separator} of~$G$
if it is a minimal $(s,t)$-separator for some $s,t \in V$.
\end{defn}

For example, let $G = (V,E)$ be
the graph  defined by 
$$V =  \{1,2,3,4,5\}, \quad \mbox{and} \quad E = \{\{1,2\},\{2,3\},\{3,4\},\{4,1\},\{1,5\}\}.$$
$G$ is a 
four-edge cycle with a pendant vertex. 
Then $\{1,3\}$ is a minimal separator of~$G$
since it is a minimal $(2,4)$-separator.

We have already seen that algorithms for counting and approximately counting
minimal separators
are useful in algorithmic applications.
There is also lots of existing work on listing minimal separators. 
Given a graph~$G$, let $n$ be the number
of vertices and let $m$ be the number of edges.
Kloks and Kratsch, and independently, Sheng and Liang, showed how to compute all $(s,t)$-minimal separators in $O(n^3)$ time per $(s,t)$-minimal separator \cite{KK98,SL97}.
Computing all minimal separators by computing $(s,t)$-minimal separators for each possible vertex pair in this way leads to an $O(n^5)$ time per minimal separator listing algorithm.
Berry, Bordat, and Cogis~\cite{BBC00} improved this approach, computing all minimal separators in $O(n^3)$ time per minimal separator.
Each of these algorithms require storing minimal separators in an adequate data structure.
Takata's algorithm~\cite{Takata-seps} generates the set of minimal separators in $O(n^3m)$ time per minimal separator but linear space.
A graph has at most $O(1.6181^n)$ minimal separators \cite{FV12}.
We study the following computational problems, 
based on our desire to count and to approximately count minimal separators.
 
\begin{prob}\label{prob:stbms}
\myprob{\stbms}{A bipartite graph $G$ and two vertices $s,t \in V(G)$}{The number of minimal $(s,t)$-separators of $G$, which we denote by $\MS(G,s,t)$}
\end{prob}

\begin{prob}\label{prob:bms}
\myprob{\bms}{A bipartite graph $G$}{The number of minimal separators of $G$, which we denote by $\MS(G)$}
\end{prob}

Theorem~\ref{thm:ms-vertex-hard} below 
shows that both problems are   $\numP$-complete
to solve exactly and are complete for~$\numP$  with respect to approximation-preserving reductions.

Motivated by applications to phylogeny~\cite{RobPreprint}
we also consider various heuristic approximations to the minimal separator problem.
We start by defining the notion of an inclusion-minimal separator, since the number of these is a natural
lower bound for the number of minimal separators.
 
 \begin{defn} Let $G$ be a graph. A minimal separator $X$ of~$G$ is
said to be an \emph{inclusion-minimal separator} if no proper subset of $X$ is a minimal separator.
\end{defn}

In the five-vertex example above, the minimal
separator $\{1,3\}$ is 
 not an inclusion-minimal separator since $\{1\}\subset \{1,3\}$ is a minimal $(5,4)$-separator.  
 However $\{1\}$ is an inclusion-minimal separator.
We consider the following computational problem.

\begin{prob}\label{prob:bims}
\myprob{\bims}{A bipartite graph $G$}{The number of inclusion-minimal separators of $G$, which we denote by $\IMS(G)$}
\end{prob}

We also consider the problem of counting $2$-component minimal separators
since these arise in phylogenetic assembly.

\begin{prob}\label{prob:stbmstwo}
\myprob{\stbmstwo}{A bipartite graph $G$ and two vertices $s,t \in V(G)$}
{The number of minimal $(s,t)$-separators $X$ of $G$ 
such that $G-X$ has exactly two connected components}
\end{prob}

\begin{prob}\label{prob:bmstwo}
\myprob{\bmstwo}{A bipartite graph $G$}
{The number of minimal  separators $X$ of $G$ 
such that $G-X$ has exactly two connected components}
\end{prob}

Finally, our main theorem about minimal separators shows that all of these
problems are $\numP$-complete and are also complete for $\numP$ with respect
to AP-reductions.

\begin{thm}\label{thm:ms-vertex-hard}
The problems 
\stbms, \bms, \stbmstwo, \bmstwo\
and \bims are \numP-complete and are equivalent to \sat under AP-reduction.
\end{thm}

Theorem~\ref{thm:ms-vertex-hard} is proved in Section~\ref{sec:minsep}.
In order to prove  it, we first
study algorithmic problems related to other natural locally-optimal structures,
namely minimal edge-separators. These problems are interesting 
for their own sake, but they are also used in the proof of Theorem~\ref{thm:ms-vertex-hard}.
In the following definitions, $G=(V,E)$ is again a graph, and $s$ and $t$ are distinct vertices of~$G$.
$F\subseteq E$ is a set of edges of~$G$.

\begin{defn}\label{defn:edge-separators}
 $F $ is an \emph{$(s,t)$-edge separator} of $G$ if $s$ and $t$ lie in different components of $G - F$. 
 If in addition no proper subset of $F$ is an $(s,t)$-edge separator of~$G$ then we say that $F$
 is a  \emph{minimal $(s,t)$-edge separator}  of~$G$. 
\end{defn}

\begin{defn}\label{defn:vertex-separators}
 $F$ is a \emph{minimal edge separator} of~$G$ if it is a minimal $(s,t)$-edge separator for some $s,t \in V$.
\end{defn}

As the following proposition shows,
there is no need to define inclusion-minimal edge separators, since these
would be the same as minimal edge separators.
\begin{prop}\label{prop:no-bimes}
Let $G = (V,E)$ be a connected graph. An edge separator $F \subseteq E$ of $G$ is minimal if and only if no proper subset of $F$ is an edge separator of $G$.
\end{prop}
\begin{proof}
This is 
immediate 
from a slightly more general proposition,
Proposition~\ref{prop:max-minl-seps-are-cuts}, which in turn is a result of Whitney~\cite{whitney-seps}.
\end{proof}

We study the following problems, showing that they are both $\numP$-complete with respect
to AP-reductions and $\numP$-complete to compute exactly.

\begin{prob}\label{prob:stbmes}
\myprob{\stbmes}{A bipartite graph $G$ and two vertices $s,t \in V(G)$}{The number of minimal $(s,t)$-edge separators of $G$, which we denote by $\MES(G,s,t)$}
\end{prob}

\begin{prob}\label{prob:bmes}
\myprob{\bmes}{A bipartite graph $G$}{The number of minimal edge separators of $G$, which we denote by $\MES(G)$}
\end{prob}

\begin{thm}\label{thm:ms-edge-hard}
The problems \bmes and \stbmes are \numP-complete and are equivalent to \sat under AP-reduction.
\end{thm}

In addition to studying maximal independent
sets and minimal vertex and edge separators,
we study  two other structures related
to maximal independent sets in bipartite graphs.

\begin{defn}\label{defn:dom-set}
Let $G$ be a graph. We say that a set $X \subseteq V(G)$ is a \emph{dominating set} in $G$ if every vertex in $V(G) \setminus X$ sends an edge into $X$.
\end{defn}

We consider the following computational problem.

\begin{prob}\label{prob:bds}
\myprob{\bds}{A bipartite graph $G$}{The number of dominating sets in $G$}
\end{prob}

It is already known~\cite{kou} that exactly counting dominating sets in bipartite graphs is~$\numP$-complete. We show
that that approximately counting them is also complete for $\numP$ with respect to AP-reductions.

\begin{thm}\label{thm:bds-hard}
$\bds \eAP \sat$.
\end{thm}

Finally, we show that maximal independent sets in bipartite graphs 
can be represented as unions of sets, so a natural set union problem is also
$\numP$-hard with respect to AP-reductions, and so is its inverse.
To describe the problem, we  use the following notation.
Throughout the paper, we write $\N$ for the set $\{1, 2, \dots\}$ of natural numbers. For all $n \in \N$, we write $[n] = \{1, 2, \dots, n\}$.  

\begin{defn}
Let $\F \subseteq 2^{[n]}$. We define $\cup \F = \bigcup_{F \in \F} F$, $\U(\F) = \{\cup\F'\mid \F' \subseteq \F\}$, and $\U_{\F}^{-1}(F) = \{\F' \subseteq \F \mid  \cup\F' = F\}$. 
\end{defn}

For example, taking $\F = \{\{1\}, \{1,2\}, \{3,4\}\}$, we have
\begin{align*}
\cup\F &= [4],\\
\U(\F) &= \{\{1\},\{1,2\},\{3,4\},\{1,3,4\},\{1,2,3,4\}\},\\
\U_{\F}^{-1}([4]) &= \{\{\{1,2\},\{3,4\}\},
                \{\{1\},\{1,2\},\{3,4\}\}\}.
\end{align*}
Note in particular that we may have $\U_{\F}^{-1}(F) = \emptyset$.

The following theorems are proved in Section~\ref{sec:related}.

\begin{prob}\label{prob:setunion}
\myprob{\setunion}{An integer $n \in \N$ and a family of sets $\F \subseteq 2^{[n]}$}{$\SU{\F}$}
\end{prob}

\begin{thm}\label{thm:su-hard}
$\setunion \eAP \sat$.
\end{thm}

Note that the connection between the two problems driving Theorem~\ref{thm:su-hard} was already known in the context of the union-closed sets conjecture --- see Bruhn, Charbit, Schaudt and Telle~\cite{bcst-setunion}. We give an explicit proof for clarity.

\begin{prob}\label{prob:unionrep}
\myprob{\unionrep}{An integer $n \in \N$ and a family of sets $\F \subseteq 2^{[n]}$}{$\UR{[n]}$}
\end{prob}

\begin{thm}\label{thm:ur-hard}
$\unionrep \eAP \sat$.
\end{thm}

\section{Preliminaries}\label{sec:notation}

Let $X$ and $Y$ be sets. Then we write $X \subseteq Y$ if $X$ is a subset of $Y$, and $X \subset Y$ if $X$ is a proper subset of $Y$. We write $2^X$ for the power set of $X$. For $t \in \N$, we write $X^{(t)}$ for the set of subsets of $X$ of cardinality $t$.

Let $X$ and $Y$ be multisets. We write $X \uplus Y$ for the disjoint union of $X$ and $Y$. We also adopt the convention that elements of a multiset with the same name are nevertheless distinguishable. 

We require our graphs to be simple, i.e. to have no loops or multiple edges. We require our multigraphs to have no loops. Let $G = (V,E)$ be a multigraph. For all $v \in V$, we write $N(v) = \{w \in V: \{v,w\} \in E\}$. For all $S \subseteq V$, we write $N(S) = \bigcup_{v \in S} N(v)$. We define the underlying graph of $G$ to be the graph with vertex set $V$ and edge set $\{e: e \in  E\}$.

Let $G = (V,E)$ be a graph. If $F \subseteq E$, we write $G-F$ for the graph $(V, E \setminus F)$. If $X \subseteq V$, we write $G-X$ for the graph $G[V \setminus X]$ induced by $G$ on $V \setminus X$.

The following definitions are standard in the field, and have been taken largely from~\cite{gj-notation}. 
We require our problem inputs to be given as finite binary strings, and write $\Sigma^*$ for the set of all such strings. A \emph{randomised approximation scheme} is an algorithm for approximately computing the value of a function $f:\Sigma^* \rightarrow \N$. The approximation scheme has a parameter $\varepsilon \in (0,1)$ which specifies the error tolerance. A \emph{randomised approximation scheme} for $f$ is a randomised algorithm that takes as input an instance $x \in \Sigma^*$ (e.g. an encoding of the graph $G$ in an instance of \maxlbis) and a rational error tolerance $\varepsilon \in (0,1)$, and outputs a rational number $z$ (a random variable depending on the ``coin tosses'' made by the algorithm) such that, for every instance $x$, $\pr(e^{-\varepsilon} f(x) \le z \le e^{\varepsilon} f(x)) \ge \frac{3}{4}$. The randomised approximation scheme is said to be a \emph{fully polynomial randomised approximation scheme}, or \emph{FPRAS}, if it runs in time bounded by a polynomial in $|x|$ and $\varepsilon^{-1}$.

Our main tool for understanding the relative difficulty of approximation counting problems is \emph{approximation-preserving reductions}. We use the notion of AP-reduction from Dyer et al.~\cite{dggj-approx}. Suppose that $f$ and $g$ are functions from $\Sigma^*$ to $\N$. An AP-reduction from $f$ to $g$ gives a way to turn an FPRAS for $g$ into an FPRAS for $f$. An \emph{approximation-preserving reduction} or \emph{AP-reduction} from $f$ to $g$ is a randomised algorithm $\mathcal{A}$ for computing $f$ using an oracle for $g$. The algorithm $\mathcal{A}$ takes as input a pair $(x, \varepsilon) \in \Sigma^* \times (0,1)$, and satisfies the following three conditions: (i) every oracle call made by $\mathcal{A}$ is of the form $(w, \delta)$, where $w \in \Sigma^*$ is an instance of $g$, and $\delta \in (0,1)$ is an error bound satisfying $\delta^{-1} \le \textnormal{poly}(|x|,\varepsilon^{-1})$; (ii) the algorithm $\mathcal{A}$ meets the specification for being a randomised approximation scheme for $f$ (as described above) whenever the oracle meets the specification for being a randomised approximation scheme for $g$; and (iii) the run-time of $\mathcal{A}$ is polynomial in $|x|$ and $\varepsilon^{-1}$ and the bit-size of the values returned by the oracle.

If an AP-reduction from $f$ to $g$ exists we write $f \lAP g$, and say that \emph{$f$ is AP-reducible to $g$}. Note that if $f \lAP g$ and $g$ has an FPRAS then $f$ has an FPRAS. (The definition of AP-reduction was chosen to make this true.) If $f \lAP g$ and $g \lAP f$ then we say that \emph{$f$ and $g$ are equivalent under AP-reduction}, and write $f \eAP g$. A word of warning about terminology: the notation $\lAP$ has been used (see e.g.~\cite{crescenzi-badap}) to denote a different type of approximation-preserving reduction which applies to optimisation problems. We will not study optimisation problems in this paper, so hopefully this will not cause confusion.

Dyer et al.~\cite{dggj-approx} studied counting problems in \numP and identified three classes of counting problems that are interreducible under AP-reductions. The first class, containing the problems that have an FPRAS, are trivially equivalent under AP-reduction since all the work can be embedded into the reduction (which declines to use the oracle). The second class is the set of problems that are equivalent to \sat, the problem of counting satisfying assignments to a Boolean formula in CNF, under AP-reduction. 
These problems are complete for $\numP$ with respect to AP-reductions.
Zuckerman~\cite{zuckerman-sat} has shown that \sat cannot have an FPRAS unless $\textnormal{RP} = \textnormal{NP}$. The same is obviously true of any problem to which \sat is AP-reducible.

The third class appears to be of intermediate complexity. It contains all of the counting problems expressible in a certain logically-defined complexity class, \rhp. Typical complete problems include counting the downsets in a partially ordered set~\cite{dggj-approx}, computing the partition function of the ferromagnetic Ising model with local external magnetic fields~\cite{gj-ferroising}, and counting the independent sets in a bipartite graph, which is formally defined as follows.

\begin{prob}\label{prob:bis}
\myprob{\bis}{A bipartite graph $G$}{The number of independent sets in $G$, which we denote by $\IS(G)$}
\end{prob}

In~\cite{dggj-approx} it was shown that \bis is complete for the logically-defined complexity class \rhp with respect to AP-reductions. Goldberg and Jerrum~\cite{gj-bisconjecture} have conjectured that there is no FPRAS for \bis.
Early indications point to the fact that it may be of intermediate complexity,
between the FPRASable problems and those that are complete for $\numP$ with respect to AP-reductions.

\section{Hardness of \maxlbis} \label{sec:maxl-bis}
 
We first prove    that \maxlbis is 
complete for $\numP$ with respect
to AP-reductions. We reduce from the well-known problem of counting independent sets in an arbitrary 
graph.

\begin{prob}\label{prob:is}
\myprob{\is} {A graph $G$}{The number of independent sets in $G$}
\end{prob}

Note that \is is complete for~$\numP$ with respect to AP-reductions ---
 indeed, the following appears as Theorem~3 of Dyer, Goldberg, Greenhill and Jerrum~\cite{dggj-approx}.

\begin{thm}\label{thm:is-hard} (DGGJ)
$\is \eAP \sat$.
\end{thm}

We can now prove Theorem~\ref{thm:maxl-bis-hard}.

{\renewcommand{\thethm}{\ref{thm:maxl-bis-hard}}
\begin{thm}
$\maxlbis \eAP \sat$.
\end{thm}}

\begin{figure}[t]
\begin{center}
\begin{tikzpicture}[scale=0.8]
\node (Ga) at (0,4) [vertex, label = above left:$1$]  {};
\node (Gb) at (4,4) [vertex, label = above right:$2$] {};
\node (Gc) at (4,0) [vertex, label = below right:$3$] {};
\node (Gd) at (0,0) [vertex, label = below left:$4$]  {};
\draw (Ga) -- (Gb) -- (Gc) -- (Gd) -- (Ga) -- (Gc);

\node (G'a) at (8,4)  [vertex, label = left:$1$]   {};
\node (G'b) at (12,4) [vertex, label = right: $2$] {};
\node (G'c) at (12,0) [vertex, label = right: $3$] {};
\node (G'd) at (8,0)  [vertex, label = left: $4$]  {};

\node (G'a') at (6.5,5.5)   [vertex, color=blue, label = above left:$v_1$]  {};
\node (G'b') at (13.5,5.5)  [vertex, color=blue, label = above right:$v_2$] {};
\node (G'c') at (13.5,-1.5) [vertex, color=blue, label = below right:$v_3$] {};
\node (G'd') at (6.5,-1.5)  [vertex, color=blue, label = below left:$v_4$]  {};

\def \x {2.5mm}
\node (G'ab11) at (9,4)  [smallvertex] {};
\node (G'ab21) at (10,4) [smallvertex] {};
\node (G'ab31) at (11,4) [smallvertex] {};
\node (G'ab13) at (9,5)  [smallvertex] {};
\node (G'ab23) at (10,5) [smallvertex] {};
\node (G'ab33) at (11,5) [smallvertex] {};
\node (G'ab12) at ($(G'ab13)!\x!(G'ab11)$) [smallvertex] {};
\node (G'ab22) at ($(G'ab23)!\x!(G'ab21)$) [smallvertex] {};
\node (G'ab32) at ($(G'ab33)!\x!(G'ab31)$) [smallvertex] {};

\node (G'bc11) at (12,3) [smallvertex] {};
\node (G'bc21) at (12,2) [smallvertex] {};
\node (G'bc31) at (12,1) [smallvertex] {};
\node (G'bc13) at (13,3) [smallvertex] {};
\node (G'bc23) at (13,2) [smallvertex] {};
\node (G'bc33) at (13,1) [smallvertex] {};
\node (G'bc12) at ($(G'bc13)!\x!(G'bc11)$) [smallvertex] {};
\node (G'bc22) at ($(G'bc23)!\x!(G'bc21)$) [smallvertex] {};
\node (G'bc32) at ($(G'bc33)!\x!(G'bc31)$) [smallvertex] {};

\node (G'cd11) at (11,0)  [smallvertex] {};
\node (G'cd21) at (10,0)  [smallvertex] {};
\node (G'cd31) at (9,0)   [smallvertex] {};
\node (G'cd13) at (11,-1) [smallvertex] {};
\node (G'cd23) at (10,-1) [smallvertex] {};
\node (G'cd33) at (9,-1)  [smallvertex] {};
\node (G'cd12) at ($(G'cd13)!\x!(G'cd11)$) [smallvertex] {};
\node (G'cd22) at ($(G'cd23)!\x!(G'cd21)$) [smallvertex] {};
\node (G'cd32) at ($(G'cd33)!\x!(G'cd31)$) [smallvertex] {};

\node (G'da11) at (8,1) [smallvertex] {};
\node (G'da21) at (8,2) [smallvertex] {};
\node (G'da31) at (8,3) [smallvertex] {};
\node (G'da13) at (7,1) [smallvertex] {};
\node (G'da23) at (7,2) [smallvertex] {};
\node (G'da33) at (7,3) [smallvertex] {};
\node (G'da12) at ($(G'da13)!\x!(G'da11)$) [smallvertex] {};
\node (G'da22) at ($(G'da23)!\x!(G'da21)$) [smallvertex] {};
\node (G'da32) at ($(G'da33)!\x!(G'da31)$) [smallvertex] {};

\def \y {4mm}
\node (G'ac11) at ($(G'a)!0.25!(G'c)!\y!270:(G'c)$) [smallvertex] {};
\node (G'ac21) at ($(G'a)!0.5 !(G'c)!\y!270:(G'c)$) [smallvertex] {};
\node (G'ac31) at ($(G'a)!0.75!(G'c)!\y!270:(G'c)$) [smallvertex] {};

\node (G'ac12) at ($(G'a)!0.25!(G'c)!\y!90:(G'c)$) [smallvertex] {};
\node (G'ac22) at ($(G'a)!0.5 !(G'c)!\y!90:(G'c)$) [smallvertex] {};
\node (G'ac32) at ($(G'a)!0.75!(G'c)!\y!90:(G'c)$) [smallvertex] {};

\draw (G'a) -- (G'ab11) -- (G'ab21) -- (G'ab31) -- (G'b);
\draw (G'a) -- (G'ab12) -- (G'ab22) -- (G'ab32) -- (G'b);
\draw (G'a) -- (G'ab13) -- (G'ab23) -- (G'ab33) -- (G'b);

\draw (G'b) -- (G'bc11) -- (G'bc21) -- (G'bc31) -- (G'c);
\draw (G'b) -- (G'bc12) -- (G'bc22) -- (G'bc32) -- (G'c);
\draw (G'b) -- (G'bc13) -- (G'bc23) -- (G'bc33) -- (G'c);

\draw (G'c) -- (G'cd11) -- (G'cd21) -- (G'cd31) -- (G'd);
\draw (G'c) -- (G'cd12) -- (G'cd22) -- (G'cd32) -- (G'd);
\draw (G'c) -- (G'cd13) -- (G'cd23) -- (G'cd33) -- (G'd);

\draw (G'd) -- (G'da11) -- (G'da21) -- (G'da31) -- (G'a);
\draw (G'd) -- (G'da12) -- (G'da22) -- (G'da32) -- (G'a);
\draw (G'd) -- (G'da13) -- (G'da23) -- (G'da33) -- (G'a);

\draw (G'a) -- (G'ac11) -- (G'ac21) -- (G'ac31) -- (G'c);
\draw (G'a) -- (G'ac12) -- (G'ac22) -- (G'ac32) -- (G'c);

\draw [color=blue] (G'a) -- (G'a');
\draw [color=blue] (G'b) -- (G'b');
\draw [color=blue] (G'c) -- (G'c');
\draw [color=blue] (G'd) -- (G'd');

\def \ew {1mm}
\node (G'abdot2) at ($(G'ab21)!0.5!(G'ab22)$)   [ellipses] {};
\node (G'abdot1) at ($(G'abdot2)!\ew!(G'ab21)$) [ellipses] {};
\node (G'abdot3) at ($(G'abdot2)!\ew!(G'ab22)$) [ellipses] {};

\node (G'bcdot2) at ($(G'bc21)!0.5!(G'bc22)$)   [ellipses] {};
\node (G'bcdot1) at ($(G'bcdot2)!\ew!(G'bc21)$) [ellipses] {};
\node (G'bcdot3) at ($(G'bcdot2)!\ew!(G'bc22)$) [ellipses] {};

\node (G'cddot2) at ($(G'cd21)!0.5!(G'cd22)$)   [ellipses] {};
\node (G'cddot1) at ($(G'cddot2)!\ew!(G'cd21)$) [ellipses] {};
\node (G'cddot3) at ($(G'cddot2)!\ew!(G'cd22)$) [ellipses] {};

\node (G'dadot2) at ($(G'da21)!0.5!(G'da22)$)   [ellipses] {};
\node (G'dadot1) at ($(G'dadot2)!\ew!(G'da21)$) [ellipses] {};
\node (G'dadot3) at ($(G'dadot2)!\ew!(G'da22)$) [ellipses] {};

\node (G'acdot2) at ($(G'ac21)!0.5!(G'ac22)$)   [ellipses] {};
\node (G'acdot1) at ($(G'acdot2)!\ew!(G'ac21)$) [ellipses] {};
\node (G'acdot3) at ($(G'acdot2)!\ew!(G'ac22)$) [ellipses] {};

\node at (Gb)   [occupied] {};

\node at (G'a') [occupied] {};
\node at (G'b)  [occupied] {};
\node at (G'c') [occupied] {};
\node at (G'd') [occupied] {};

\node at (G'ab21) [smalloccupied] {};
\node at (G'ab12) [smalloccupied] {};
\node at (G'ab23) [smalloccupied] {};
\node at (G'bc21) [smalloccupied] {};
\node at (G'bc22) [smalloccupied] {};
\node at (G'bc33) [smalloccupied] {};
\node at (G'cd21) [smalloccupied] {};
\node at (G'cd22) [smalloccupied] {};
\node at (G'cd23) [smalloccupied] {};
\node at (G'da21) [smalloccupied] {};
\node at (G'da22) [smalloccupied] {};
\node at (G'da23) [smalloccupied] {};
\node at (G'ac21) [smalloccupied] {};
\node at (G'ac12) [smalloccupied] {};
\node at (G'ac32) [smalloccupied] {};

\node (Glabel)  at (2, 6.5)  [font=\Large] {$G$};
\node (G'label) at (10, 6.5) [font=\Large] {$G'$};

\end{tikzpicture}
\end{center}

\caption{An example of the reduction from an instance $G$ of \is to an instance $G'$ of \maxlbis used in the proof of Theorem~\ref{thm:maxl-bis-hard}. The boxes around vertices indicate a non-maximal independent set in $G$ and one of its maximal counterparts in $G'$. Note in particular how the presence of $v_4$ ensures that vertex 4 has an occupied neighbour in $G'$.}
\label{fig:maxlbis}
\end{figure}
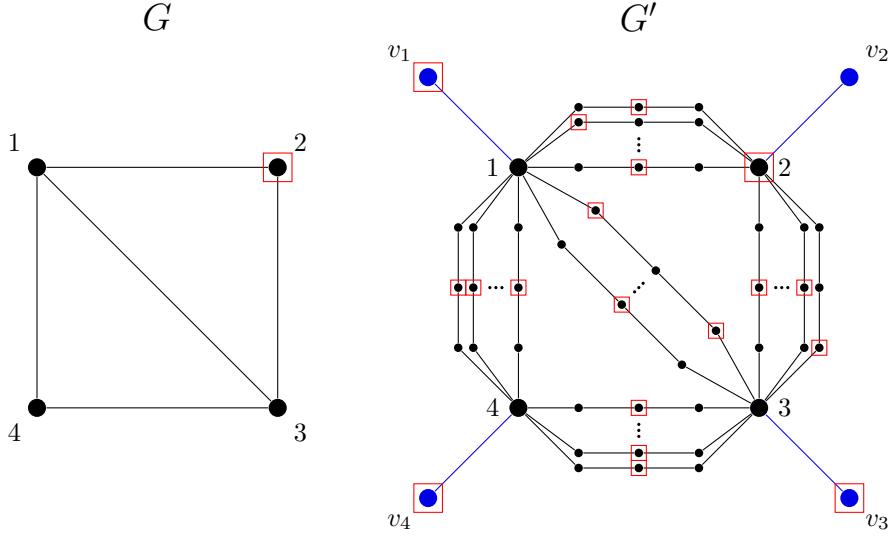

\begin{proof}
Since $\maxlbis$ is in $\numP$,  $\maxlbis \lAP \sat$
follows from \cite{dggj-approx}. To go the other direction, we will show $\is \lAP \maxlbis$.
Let $\MIS(G)$ denote the number of maximal independent sets in a graph~$G$.
Let $G=(V,E)$ be an instance of \is. Without loss of generality let $V  = [n]$ for some $n \in \N$, let $m = |E|$, and let $t = n+2$. We shall construct an instance $G'$ of \maxlbis with the property that
 $
\IS(G) \le \MIS(G')/2^{tm} \le \IS(G) + \frac{1}{4}$,
which  will be sufficient for the reduction. See Figure~\ref{fig:maxlbis} for an example.

Informally, we obtain a bipartite graph~$G'$ 
(an instance of \maxlbis) from $G$ by first $t$-thickening and 
then $4$-stretching each of $G$'s edges 
and by also
adding a bristle to each of $G$'s vertices. Formally, we define $G'$ as follows. For each 
$e \in E$ let $X_e$, $Y_e$ and $Z_e$ be sets of $t$ vertices. We require all of these sets to be disjoint from each other and from $[n]$.
Write $X_e = \{x_{e}^k  \mid k \in [t]\}$, $Y_{e} =\{y_{e}^k \mid  k \in [t]\}$, and $Z_{e} =\{z_{e}^k \mid k \in [t]\}$.
Also, let $W = \bigcup_{e \in E} X_e \cup Y_e \cup Z_e$.
Let $V^*=\{v_1, \dots, v_n\}$ be 
a set of distinct vertices which is
disjoint from $[n] \cup   W$. Then we define 
\begin{align*}
V(G') & = [n] \cup V^*
            \cup W, \\
E(G') & =  \{\{i,v_i\}\mid i \in [n] \} \cup \bigcup_{\substack{e=\{i,j\} \in E \\ i < j \\ k \in [t]}} 
                \{\{i,x_e^k\}, \{x_e^k,y_e^k\}, \{y_e^k,z_e^k\},
                \{z_e^k,j\}\}.
\end{align*}

Let $S \subseteq [n]$ be an arbitrary set. We shall determine the number $\MIS_S(G')$ of maximal independent sets $T \subseteq V(G')$ with $T \cap [n] = S$, and thereby bound $\MIS(G')$.
 
First, note that for every $S \subseteq [n]$, the set
$S \cup \{v_i \in V^* \mid i\not\in S\} \cup \bigcup_e Y_e$
is a maximal independent set of~$G'$, so
$\MIS_S(G')$ is non-zero.
Also, if $T$ is a maximal independent set of~$G'$ and $T\cap [n]=S$ then
$T\cap V^* = \{v_i \in V^* \mid i\not\in S\}$. In particular, this implies that every unoccupied vertex in $[n]$ has an occupied neighbour in $V^*$.
  
Consider an edge $e=\{i,j\}\in E$, where $i < j$, and a value $k\in[t]$.
If $T$ is a maximal independent set of~$G'$ containing both~$i$
and~$j$ then $T \cap  \{x_e^k,y_e^k,z_e^k\} = \{y_e^k\}$.
On the other hand, if $T$ is a maximal independent set of~$G'$ 
containing~$i$ but not~$j$
then $T\cap \{x_e^k,y_e^k,z_e^k\}$ can either be $\{y_e^k\}$ or $\{z_e^k\}$.
This choice can be made independently for each~$k\in[t]$.
Similarly, if $T$ is a maximal independent set of~$G'$
containing neither of~$i$ and~$j$ then
$T \cap \{x_e^k,y_e^k,z_e^k\}$ can either be $\{x_e^k,z_e^k\}$, or $\{y_e^k\}$.

Given $S\subseteq [n]$, let $\mu(S)$ be the number of edges of~$G$
with both endpoints in~$S$.
We conclude from the previous observations that  $\MIS_S(G') =  2^{(m - \mu(S)) t}$
so $\MIS(G') = \sum_{S\subseteq [n]} 2^{(m - \mu(S)) t}$.
Since each independent set $S$ of~$G$
has $\mu(S)=0$, 
$\MIS(G') \geq \IS(G) 2^{mt}$.
Furthermore,
since there are at most $2^n$ sets $S\subseteq[n]$ that are not independent sets of~$G$,
and each of these has $\mu(S)\geq 1$, we have

\begin{equation}
\label{target} 
\IS(G) \le \frac{\MIS(G')}{2^{tm}} \le \IS(G) + 2^n  2^{-t} = \IS(G) + \frac{1}{4}.
\end{equation}
Equation~(\ref{target}) implies that there is an AP-reduction from \is\ to \maxlbis.
The details of the reduction showing
how to tune the accuracy parameter 
in the oracle call for approximating $\MIS(G')$ in order to get a sufficiently good 
approximation to $\IS(G)$  are exactly as in the proof of Theorem~3 of~\cite{dggj-approx}. 

\end{proof}

\section{Minimal separator problems}\label{sec:minsep}

\subsection{Two intermediate problems}

In this section, we shall present hardness proofs for two intermediate problems. We will then subsequently use these problems as reduction targets in our proofs of Theorems~\ref{thm:ms-vertex-hard}
and~\ref{thm:ms-edge-hard}. We first explicitly generalise Definitions~\ref{defn:edge-separators} and~\ref{defn:vertex-separators} to multigraphs in the natural way. We avoided doing so in the introduction because the graph separator problems that we have previously defined are trivially equivalent to their multigraph variants --- we will only use these definitions for intermediate problems.

\begin{defn}
Let $G = (V,E)$ be a multigraph, and let $s,t \in V(G)$. A multiset $F \subseteq E$ is an \emph{$(s,t)$-edge separator} of $G$ if $s$ and $t$ lie in different components of $G - F$. We say $F$ is a \emph{minimal $(s,t)$-edge separator} if no proper submultiset of $F$ is an $(s,t)$-edge separator, and write $\MES(G,s,t)$ for the number of minimal $(s,t)$-edge separators of $G$. 
\end{defn}

\begin{defn}
Let $G = (V,E)$ be a multigraph, and let $F \subseteq E$. We say $F$ is a \emph{minimal edge separator} if it is a minimal $(s,t)$-edge separator for some $s,t \in V$, and write $\MES(G)$ for the number of minimal edge separators of $G$.
\end{defn}

We now define our two intermediate problems.

\begin{prob}\label{prob:largemes}
\myprob{\largemes}{A multigraph $G$ and the maximum cardinality $x$ of any minimal edge separator in $G$}{The number of minimal edge separators of $G$ with maximum cardinality, which we denote by $\TMES(G)$}
\end{prob}

\begin{prob}\label{prob:largestmes}
\myprob{\largestmes}{A multigraph $G$, two distinct vertices $s,t \in V$, and the maximum cardinality $y$ of any minimal $(s,t)$-edge separator in $G$}{The number of minimal $(s,t)$-edge separators of $G$ with maximum cardinality, which we denote by $\TMES(G,s,t)$}
\end{prob}

Note that the input restrictions in the definitions of \largemes and \largestmes are motivated purely by their uses as intermediate problems in reductions. When we use them, we will be able to prove that their respective promises are satisfied. As the next proposition shows, both \largemes and \largestmes can be expressed in terms of vertex cuts. It is a widely known result and was first proved by Whitney~\cite{whitney-seps} --- we give a proof here for completeness.

\begin{prop}\label{prop:max-minl-seps-are-cuts}
Let $G = (V,E)$ be a connected multigraph. Then a multiset $F \subseteq E$ is a minimal edge separator of $G$ if and only if $G-F$ has exactly two non-empty components, and $F$ is the multiset of edges between them.
\end{prop}
\begin{proof}
For any non-empty set $S \subset V$ such that $G[S]$ and $G[V \setminus S]$ are connected, taking an arbitrary $s \in S$ and $t \in V \setminus S$, it is immediate that the multiset of edges between $S$ and $V \setminus S$ is a minimal $(s,t)$-edge separator and hence a minimal edge separator. 

Conversely, let $F \subseteq E(G)$ be a minimal $(s,t)$-edge separator for some $s,t \in V$. Suppose $G-F$ has (at least) three components $C_1$, $C_2$ and $C_3$. Without loss of generality, suppose $s \in C_1$ and $t \in C_2$. Then since $G$ is connected, $F$ must contain an edge $e$ from $C_1 \cup C_2$ to $C_3$. But then $F \setminus \{e\}$ is still an $(s,t)$-edge separator, contradicting minimality. Hence $G-F$ has only two components, as required.
\end{proof}

Thus we may view counting maximum minimal edge separators as counting maximum vertex cuts subject to the requirement that each part of the vertex cut is connected. We shall therefore prove hardness for \largemes and \largestmes by adapting a folklore proof that MAX-CUT is NP-complete (see e.g. Exercise 7.25 of Sipser~\cite{sipser}).  The original proof works by reduction from 3-NAE-SAT -- we shall instead reduce from the following variant of the problem.

\begin{defn}
We define \nae to be a logical clause as follows. Let $x_1$, $x_2$ and $x_3$ be literals, and let $\sigma:\{x_1, x_2, x_3\} \rightarrow \{0,1\}$ be a truth assignment. Then under $\sigma$, 
\[\nae(x_1, x_2, x_3) = \begin{cases}0 & \textnormal{ if } \sigma(x_1) = \sigma(x_2) = \sigma(x_3) \\ 1 & \textnormal{ otherwise}.\end{cases}\]
\end{defn}
\begin{defn}
We define a \emph{monotone \tnae} formula $\phi$ to be any logical formula of the form $\bigwedge_{i \in [k]}\C_i$, where $k \in \N$ and $\C_1, \dots, \C_k$ are \nae clauses containing three distinct and un-negated literals, e.g.\ three distinct variables.
\end{defn}

\begin{prob}\label{prob:tnaesat}
\myprob{\tnaesat}{A satisfiable monotone \tnae formula $\phi$}{The number of satisfying assignments of $\phi$}
\end{prob}

We first prove hardness for \tnaesat by reduction from \is.

\begin{lem}\label{lem:tnaesat-hard}
\tnaesat is \sat-hard to approximate and is \numP-complete.
\end{lem}
\begin{proof}
For every instance $\phi$ of \tnaesat, let $\SAT(\phi)$ be the number of satisfying assignments of $\phi$. Since \tnaesat is in \numP, we have $\tnaesat \lAP \sat$ by~\cite{dggj-approx}. Let $G=(V,E)$ be an instance of \is, which is hard by Theorem~\ref{thm:is-hard}. We shall construct an instance $\phi$ of \tnaesat with the property that $\SAT(\phi) = 2 \cdot \IS(G)$, from which the result follows immediately.

We identify $V$ with a set of logical variables. Let $x$ be a new variable distinct from the variables in $V$.  Then we define
\[\phi = \bigwedge_{\{i,j\} \in E} \nae(i,j,x).\]
Note that $\phi$ is satisfiable by setting $x$ to 1 and all other variables to 0, so $\phi$ is an instance of \tnaesat. 

Suppose $\sigma:V\cup\{x\} \rightarrow \{0,1\}$ is a satisfying assignment of $\phi$. Then we may define an independent set $S$ as follows.
\[S = \{v \in V \mid \sigma(v) = \sigma(x)\}.\]
Since $\sigma$ is a satisfying assignment, we cannot have $\sigma(i) = \sigma(j) = \sigma(x)$ for any $\{i,j\} \in E$, and so $S$ is an independent set.

Conversely, suppose $S$ is an independent set of $G$ and let $1_S$ be the indicator function of $S$. Then $S$ corresponds to two satisfying assignments $\sigma_0, \sigma_1:V\cup \{x\}\rightarrow \{0,1\}$ of $\phi$. Indeed, let $\sigma_1(x) = 1$, and let $\sigma_1(v) = 1_S(v)$ for all $v \in V$. Then $\sigma_1$ satisfies every clause $\nae(i,j,x)$ of $\phi$, since $\sigma_1(x) = 1$ and at most one of $i$ and $j$ lies in $S$. We then define $\sigma_0 = 1 - \sigma_1$, which is a satisfying assigmnent since $\sigma_1$ is a satisfying assignment.

Thus each satisfying assignment of $\phi$ corresponds to a unique independent set of $G$, and each independent set of $G$ corresponds to exactly two satisfying assignments of $\phi$. The result therefore follows.
\end{proof}

We now reduce \tnaesat to \largemes and \largestmes.

\begin{lem}\label{lem:largeminsep-hard}
\largemes and \largestmes are \sat-hard to approximate and are \numP-complete.
\end{lem}

\begin{figure}
\begin{center}
\begin{tikzpicture}[scale=0.75]

\matrix[matrix of nodes, inner sep=0pt, column sep=0pt] at (0, 0) {
      \node (phistart) {\strut NAE$($}; &
      \node (phi1x1)   {\strut $x_1$};  &
      \node (phi1x1x2) {\strut $,$};    &
      \node (phi1x2)   {\strut $x_2$};  &
      \node (phi1x2x3) {\strut $,$};    &   
      \node (phi1x3)   {\strut $x_3$};  &
      \node (phi12)    {\strut $)\wedge \textnormal{NAE}($}; &
      \node (phi2x3)   {\strut $x_3$};  &
      \node (phi2x3x4) {\strut $,$};    &
      \node (phi2x4)   {\strut $x_4$};  &
      \node (phi2x4x5) {\strut $,$};    &   
      \node (phi2x5)   {\strut $x_5$};  &      
      \node (phiend)   {\strut $)$}; \\
};

\def \rowsepa {1.5};
\def \rowsepb {0.5};
\node (sig11) at ($(phi1x1) - (0,\rowsepa)$) {\strut 1};
\node (sig12) at ($(phi1x2) - (0,\rowsepa)$) {\strut 1};
\node (sig13) at ($(phi1x3) - (0,\rowsepa)$) {\strut 0};
\node (sig14) at ($(phi2x3) - (0,\rowsepa)$) {\strut 0};
\node (sig15) at ($(phi2x4) - (0,\rowsepa)$) {\strut 1};
\node (sig16) at ($(phi2x5) - (0,\rowsepa)$) {\strut 0};

\node (sig21) at ($(phi1x1) - (0,\rowsepa+\rowsepb)$) {\strut 0};
\node (sig22) at ($(phi1x2) - (0,\rowsepa+\rowsepb)$) {\strut 0};
\node (sig23) at ($(phi1x3) - (0,\rowsepa+\rowsepb)$) {\strut 1};
\node (sig24) at ($(phi2x3) - (0,\rowsepa+\rowsepb)$) {\strut 1};
\node (sig25) at ($(phi2x4) - (0,\rowsepa+\rowsepb)$) {\strut 0};
\node (sig26) at ($(phi2x5) - (0,\rowsepa+\rowsepb)$) {\strut 1};

\node (sig1) at ($(sig11) - (0.75,0)$) {\strut $\sigma_S:$};
\node (sig2) at ($(sig21) - (0.75,0)$) {\strut $\sigma_{\overline{S}}:$};

\draw [->] ($(sig11)!0.3!(phi1x1)$) -- ($(sig11)!0.7!(phi1x1)$);
\draw [->] ($(sig12)!0.3!(phi1x2)$) -- ($(sig12)!0.7!(phi1x2)$);
\draw [->] ($(sig13)!0.3!(phi1x3)$) -- ($(sig13)!0.7!(phi1x3)$);
\draw [->] ($(sig14)!0.3!(phi2x3)$) -- ($(sig14)!0.7!(phi2x3)$);
\draw [->] ($(sig15)!0.3!(phi2x4)$) -- ($(sig15)!0.7!(phi2x4)$);
\draw [->] ($(sig16)!0.3!(phi2x5)$) -- ($(sig16)!0.7!(phi2x5)$);

\foreach \i in {1, ..., 5} {
    \node (Gtopx\i) at ($(3.5 + 1.5*\i,0)$)  [vertex] {};
    \node (Gbotx\i) at ($(3.5 + 1.5*\i,-2)$) [vertex] {};
}

\foreach \i in {1, ..., 5} {
    \foreach \j in {1, ..., 5} {
        \draw [color=grey, line width = 1mm] (Gtopx\i) -- (Gbotx\j);
    }
}

\foreach \i in {1, ..., 5} {
	\pgfmathtruncatemacro \ii {\i+1}
    \foreach \j in {\i, ..., 5} {
    	\ifthenelse{\equal{\ii}{\j}} {
    	    \draw [color=grey, line width = 1mm] (Gtopx\i) to (Gtopx\j);
    	    \draw [color=grey, line width = 1mm] (Gbotx\i) to (Gbotx\j);
    	} {    	    	    
            \draw [color=grey, line width = 1mm] (Gtopx\i) to [bend left]  (Gtopx\j);
            \draw [color=grey, line width = 1mm] (Gbotx\i) to [bend right] (Gbotx\j);
        }
    }
}

\foreach \i in {1, ..., 5} {
	\draw [color=blue, very thin] (Gtopx\i) -- (Gbotx\i);
}

\draw [color=red, thick] (Gtopx1) -- (Gtopx2) -- (Gbotx3) -- (Gtopx1);
\draw [color=red, thick] (Gbotx3) -- (Gtopx4) -- (Gbotx5) to [bend left] (Gbotx3);

\node at (Gtopx1) [label = above:\strut $x_1$]            {};
\node at (Gtopx2) [label = above:\strut $x_2$]            {};
\node at (Gtopx3) [label = above:\strut $\overline{x_3}$] {};
\node at (Gtopx4) [label = above:\strut $x_4$]            {};
\node at (Gtopx5) [label = above:\strut $\overline{x_5}$] {};

\node at (Gbotx1) [label={[label distance = 6]270:\strut $\overline{x_1}$}] {};
\node at (Gbotx2) [label={[label distance = 6]270:\strut $\overline{x_2}$}] {};
\node at (Gbotx3) [label={[label distance = 6]270:\strut $x_3$}]            {};
\node at (Gbotx4) [label={[label distance = 6]270:\strut $\overline{x_4}$}] {};
\node at (Gbotx5) [label={[label distance = 6]270:\strut $x_5$}]            {};

\draw [thick, dashed] (4,-1) -- (12.2,-1);
\node at (11.75, -0.5) {$\strut S$};
\node at (11.75, -1.5) {$\strut \overline{S}$};

\node at (0,1.5) [font = \Large] {$\phi$};
\node at ($(Gtopx3) + (0,1.5)$) [font = \Large] {$G$};
\end{tikzpicture}
\end{center}

\caption{An example of the reduction from an instance $\phi$ of \tnaesat to an instance $(G,k)$ of \largemes used in the proof of Lemma~\ref{lem:largeminsep-hard}. The thin blue edges of $G$ are elements of $F_1$, the thick red edges are elements of $F_2$, and the very thick grey edges are elements of $F_3$. In this example we have $k = {\color{blue}{5}} + {\color{red}{4}} + {\color{darkgrey}{5\cdot 5}} = 34$, and a minimal edge separator is maximum if and only if it contains all edges of $F_1$ and 4 edges of $F_2$.} 
\label{fig:largeminsep}
\end{figure}
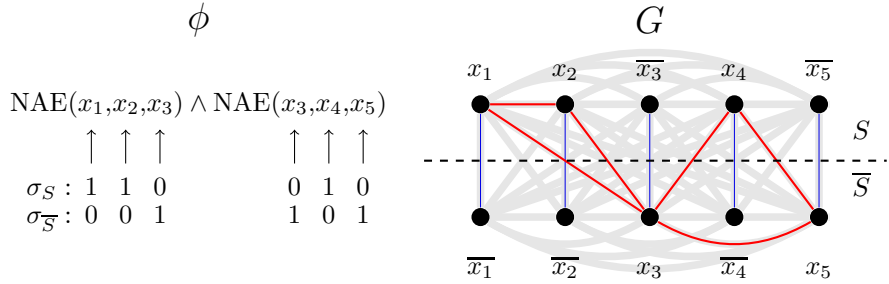

\begin{proof}
Since \largemes and \largestmes are in \numP, it follows that $\largemes \lAP \sat$ and
$\text{\#}(s,t)\text{-Large}\-\text{MinimalEdgeSeps}\lAP \sat$ by~\cite{dggj-approx}. We will first prove the result for \largemes. Let $\phi$ be an instance of \tnaesat, which is hard by Lemma~\ref{lem:tnaesat-hard}. Let $x_1, \dots, x_n$ be the variables of $\phi$, and let $\C_1, \dots, \C_m$ be the clauses of $\phi$. We shall construct an instance $(G,k)$ of \largemes with the property that $\SAT(\phi) = 2 \cdot \TMES(G)$, from which the result follows immediately. See Figure~\ref{fig:largeminsep} for an example.

We define $G = (V,E)$ as follows. Let 
\[V = \{x_i\mid  i \in [n]\} \cup \{\overline{x_i}\mid i \in [n]\}.\]
Let $C_i \subseteq V$ be the set of variables appearing in clause $\C_i$. We now define sets of edges
\begin{align*}
F_1 &= \{\{x_i, \overline{x_i}\}\mid  i \in [n]\}, \\
F_2 &= \biguplus_{i \in [m]} C_i^{(2)}, \\
F_3 &= V^{(2)}.
\end{align*}
We then define $E = F_1 \uplus F_2 \uplus F_3$. Finally, let $k = n + 2m + n^2$.

Suppose that $F$ is a minimal edge separator of $G$. By Proposition~\ref{prop:max-minl-seps-are-cuts}, $G-F$ has exactly two components $S$ and $V \setminus S$. We claim that $|F| \le k$, with equality if and only if the following properties hold.
\begin{enumerate}
\item For all $i \in [n]$, $|\{x_i, \overline{x_i}\} \cap S| = 1$.
\item For all $i \in [m]$, $|F \cap C_i^{(2)}| = 2$.
\end{enumerate}

First, note that $|F \cap F_1| \le n$ with equality if and only if (i) holds. Second, note that for all $i\in [m]$, we have $|F \cap C_i^{(2)}| \le 2$ with equality for all $i$ if and only if (ii) holds. Finally, note that
\[|F \cap F_3| = |S|(2n-|S|) = n^2 - (n-|S|)^2 \le n^2\]
with equality if and only if $|S| = n$ (which is implied by (i)). Hence
\[|F| = |F \cap F_1| + \sum_{i \in [m]}|F \cap C_i^{(2)}| + |F \cap F_3| \le k,\]
with equality if and only if (i) and (ii) hold. We will soon see that satisfying assignments of $\phi$ correspond to minimal edge separators satisfying (i) and (ii). Since $\phi$ is satisfiable, this will imply in particular that $(G,k)$ is an instance of \largemes.

We now define a two-to-one correspondence between satisfying assignments of $\phi$ and minimal edge separators of $G$ of cardinality $k$. Given a satisfying assignment $\sigma:\{x_1, \dots, x_n\} \rightarrow \{0,1\}$, let $S = \{x_i\mid  \sigma(x_i) = 1\} \cup \{\overline{x_i}\mid  \sigma(x_i) = 0\}$ and let $f(\sigma)$ be the multiset of edges from $S$ to $V \setminus S$. Note that since $G$ contains a spanning clique it is immediate that $f(\sigma)$ is a minimal $(x_1,\overline{x_1})$-edge separator, and hence a minimal edge separator. Moreover, since $\sigma$ is a satisfying assignment, $f(\sigma)$ satisfies (i) and (ii) and therefore has cardinality $k$. It is immediate that $f$ is a two-to-one map, with $f(\sigma) = f(1-\sigma)$. It remains only to prove that $f$ is surjective.

Let $F$ be a minimal edge separator of $G$ of cardinality $k$, let $S$ be a component of $G-F$, and let $\sigma_S:\{x_1, \dots, x_n\}\rightarrow \{0,1\}$ be given by
\[\sigma_S(x_i) = \begin{cases}1 & \textnormal{ if }x_i \in S, \\ 0 & \textnormal{ if }\overline{x_i} \in S.\end{cases}\]
By property (i), $\sigma_S$ is well-defined. Let $\C_i$ be a clause of $\phi$. Then $C_i \cap S$ is the set of literals in $\C_i$ which are true under $\sigma_S$, and so $\sigma_S$ satisfies $\C_i$ by property (ii). Hence $\sigma_S$ is a satisfying assignment of $\phi$, and so $\SAT(\phi) = 2 \cdot \TMES(G)$ as required.

Note that any maximum minimal edge separator in $G$ is a maximum minimal $(x_1,\overline{x_1})$-edge separator and vice versa, and so we also have $\SAT(\phi) = 2 \cdot \TMES(G, x_1, \overline{x_1})$. The result therefore follows for \largestmes as well.
\end{proof}

\subsection{Hardness of minimal separator problems}

The remaining reductions necessary to prove Theorems~\ref{thm:ms-vertex-hard} and~\ref{thm:ms-edge-hard} are all quite similar. For convenience, we combine them into the following two lemmas. The first lemma will be used to prove Theorem~\ref{thm:ms-edge-hard}.

\begin{lem}\label{lem:minsep-core-edge}
Let $G = (V,E)$ be a connected multigraph, writing $n = |V|$ and $m = |E|$. Suppose $(G,x)$ is an instance of \largemes, and $(G,s,t,y)$ is an instance of \largestmes. Let $k = \lceil m+\log_2(m)+10\rceil$. Then there exists a graph $G'$ such that the following properties hold.
\begin{enumerate}
\item $G'$ is bipartite, $V \subseteq V(G')$, and $|V(G')| \le |E|k+|V|$.
\item $\TMES(G) \le \MES(G')/2^{kx} \le \TMES(G) + \frac{1}{4}$.
\item $\TMES(G,s,t) \le \MES(G',s,t)/2^{ky} \le \TMES(G,s,t) + \frac{1}{4}$.
\end{enumerate}
\end{lem}
\begin{proof}
Informally, we form $G'$ by first $k$-thickening and then 2-stretching each edge of $G$. Formally, we define $G'$ as follows. For each $e \in E$ let $X_e$ be a set of $k$ vertices, disjoint from $V$, where $X_e \cap X_f = \emptyset$ whenever $e \ne f$. Let $X = \bigcup_{e \in E}X_e$. Then we define
\begin{align*}
V(G') & = V \cup X, \\
E(G') & = \bigcup_{e = \{u,v\} \in E} \{\{u,w\},\{w,v\}\mid w \in X_e\}.
\end{align*}
Thus $G'$ satisfies property (i). For each $e = \{u,v\} \in E$, let $P^e_1, \dots, P^e_k$ be the internally vertex-disjoint paths in $G'$ of the form $uwv$ with $w \in X_e$.

We say a minimal edge separator $F'$ of $G'$ is \emph{good} if it is not of the form $E(P^e_i)$ for some $e \in E$, $i \in [k]$. Note that every good minimal edge separator $F'$ of $G'$ satisfies the following properties.
\begin{enumerate}
\item[(a)] $|F' \cap E(P^e_i)| \le 1$ for all $e \in E$, $i \in [k]$. 
\item[(b)] If $|F' \cap E(P^e_i)| = 1$ for some $e \in E$, $i \in [k]$, then $|F' \cap E(P^e_j)| = 1$ for all $j \in [k]$. 
\end{enumerate}
For a good minimal edge separator $F'$ of~$G'$, write 
\[\pi(F') = \{e \in E \mid F' \cap E(P^e_i) \ne \emptyset\textnormal{ for some }i\in [k]\}.\]
We say that a minimal edge separator $F$ of $G$ corresponds to a good minimal edge separator $F'$ of $G'$ when $F = \pi(F')$. By properties (a) and (b), any minimal edge separator $F$ of $G$ corresponds to exactly $2^{k|F|}$ good minimal edge separators of $G'$. Conversely, any good minimal edge separator of $G'$ corresponds to a single minimal edge separator of $G$. Finally, there are exactly $mk$ non-good minimal edge separators of $G'$. Hence, writing $\MES_i(G)$ for the number of minimal edge separators of $G$ with cardinality $i$, we have
\[\MES(G') = \sum_{i=1}^x\MES_i(G)\cdot 2^{ki} + mk.\]

It follows immediately that $\MES(G')/2^{kx} \ge \TMES(G)$. Moreover, we have
\begin{align*}
\MES(G')/2^{kx} &=   \TMES(G) + \sum_{i=1}^{x-1} \MES_i(G)\cdot 2^{k(i-x)} + mk\cdot 2^{-kx} \\
                      &\le \TMES(G) + m\cdot 2^m \cdot 2^{-k} + k^2\cdot 2^{-k} \\
                      &\le \TMES(G) + \frac{1}{8} + \frac{1}{8} = \TMES(G) + \frac{1}{4}.
\end{align*}
(In the penultimate inequality, we use the fact that $G$ is connected and so $x \ge 1$. In the final inequality, we use the fact that $k \ge 10$ and hence $k^2\cdot 2^{-k} \le 1/8$.) Hence $G'$ satisfies property (ii). Moreover, minimal $(s,t)$-edge separators of $G$ correspond only to good minimal $(s,t)$-edge separators of $G'$ and vice versa, and so $G'$ satisfies property (iii) by the same argument.
\end{proof}

We can now prove Theorem~\ref{thm:ms-edge-hard}.

{\renewcommand{\thethm}{\ref{thm:ms-edge-hard}}
\begin{thm}
The problems \bmes and \stbmes are \numP-complete and are equivalent to \sat under AP-reduction.
\end{thm}}
\begin{proof}
Both problems are in \numP, and hence AP-reducible to \sat by~\cite{dggj-approx}. As in the proof of Theorem~\ref{thm:maxl-bis-hard}, Lemma~\ref{lem:minsep-core-edge} implies that 
\begin{align*}
\largemes &\lAP \bmes, \\
\largestmes &\lAP \stbmes.
\end{align*}
Moreover, since $\TMES(G)$ and $\TMES(G,s,t)$ are integers for all $G$, $s$ and $t$, Lemma~\ref{lem:minsep-core-edge} also yields exact Turing reductions. The result therefore follows by Lemma~\ref{lem:largeminsep-hard}.
\end{proof}

The second lemma will be used to prove Theorem~\ref{thm:ms-vertex-hard}.

\begin{lem}\label{lem:minsep-core-vertex}
Let $G = (V,E)$ be a connected multigraph, writing $n = |V|$ and $m = |E|$. Suppose $(G,x)$ is an instance of \largemes, and $(G,s,t,y)$ is an instance of \largestmes. Let $k = \lceil m+n+\log_3(n^2)+16\rceil$. Then there exists a graph $G'$ such that the following properties hold.
\begin{enumerate}
\item $G'$ is bipartite, $V \subseteq V(G')$, and $|V(G')| \le 3|E|k+|V|$.
\item $\TMES(G) \le \MS(G')/3^{kx} \le \TMES(G) + \frac{1}{4}$.
\item $\TMES(G,s,t) \le \MS(G',s,t)/3^{ky} \le \TMES(G,s,t) + \frac{1}{4}$.
\item $\TMES(G) \le \IMS(G')/3^{kx} \le \TMES(G) + \frac{1}{4}$.
\end{enumerate}
\end{lem}

\begin{proof}
Informally, we form $G'$ by first $k$-thickening and then 4-stretching each edge of $G$. Formally, for each $e \in E$, let $X^e$, $Y^e$ and $Z^e$ be sets of $k$ vertices. We require all of these sets to be disjoint from each other and from $V$. For each $e \in E$, write $X^e = \{x^e_1, \dots, x^e_k\}$, $Y^e = \{y^e_1, \dots, y^e_k\}$ and $Z^e = \{z^e_1, \dots, z^e_k\}$. Write $W^e = X^e \cup Y^e \cup Z^e$, and $W = \bigcup_{e \in E} W^e$. Arbitrarily labelling $e$'s endpoints as $u$ and $v$, for each $i \in [k]$ let $P^e_i$ be the path $ux^e_iy^e_iz^e_iv$. Thus the paths $P^e_1,\dots,P^e_k$ are $k$ internally vertex-disjoint paths of length 4 between $e$'s endpoints with $V(P^e_i) = \{u, x^e_i, y^e_i, z^e_i, v\}$. Then we define
\begin{align*}
V(G') & = V \cup W, \\
E(G') & = \bigcup_{\substack{e \in E\\ i \in [k]}} E(P^e_i).
\end{align*}
It is immediate that $G'$ satisfies property (i).

We will be able to associate minimal separators of $G'$ with minimal edge separators of $G$ in much the same way as in the proof of Lemma~\ref{lem:minsep-core-edge}, but the correspondence will be messier since a minimal separator of $G'$ may contain vertices of $V$. Indeed, there may be exponentially many such separators in $k$.

We define our correspondence as follows. If $X$ is a minimal separator of $G'$, we write \[\pi(X) = \{e \in E \mid X \cap W^e \ne \emptyset\}.\]
We say a minimal separator $X$ of $G'$ is \emph{$z$-good}, where $z \in \N$, if it satisfies the following conditions.
\begin{enumerate}
\item[(a)]We have $|X \cap V(P^e_i)| \le 1$ for all $e \in E$, $i \in [k]$.
\item[(b)]Whenever $|X \cap V(P^e_i)| = 1$ for some $e \in E$ and $i \in [k]$, we have $|X \cap V(P^e_j)| = 1$ for all $j \in [k]$.
\item[(c)]We have $X \cap V = \emptyset$.
\item[(d)]We have $|\pi(X)| = z$.
\end{enumerate}
We say that $X$ is \emph{good} if it is $z$-good for some $z \in \N$.

\begin{claim}\label{claim:core-ms-lem}
All but at most $3^{kx}/4$ minimal separators of $G'$ are $x$-good, and all but at most $3^{ky}/4$ minimal $(s,t)$-separators of $G'$ are $y$-good.
\end{claim}

We shall defer the proof of Claim~\ref{claim:core-ms-lem} for the moment. We say that each good minimal separator $X$ of $G'$ corresponds to the multiset $\pi(X) \subseteq E$. Note that any minimal edge separator $F$ of $G$ corresponds to exactly $3^{k|F|}$ good minimal separators of $G'$, all of which are $|F|$-good by the definition of $z$-goodness. Conversely each $z$-good minimal separator of $G'$ corresponds to a single multiset $F \subseteq E$, which is a minimal edge separator of $G$ with cardinality $z$. Hence by Claim~\ref{claim:core-ms-lem},
\[\TMES(G)\cdot 3^{kx}  \le \MS(G') \le \TMES(G)\cdot 3^{kx} + \frac{3^{kx}}{4}.\]
Hence (ii) is satisfied. Moreover, good minimal $(s,t)$-separators of $G'$ correspond to minimal $(s,t)$-edge separators of $G$ and vice versa, so (iii) is likewise satisfied by Claim~\ref{claim:core-ms-lem}. 

Finally, we claim that the following holds.
\begin{equation}\label{eq:two-seps}
\parbox{25em}{Every good minimal separator $X$ of $G'$ separates $G'-X$ into exactly two components.}
\end{equation}
Indeed, $\pi(X)$ is a minimal edge separator of $G$, and so by Proposition~\ref{prop:max-minl-seps-are-cuts} $G-\pi(X)$ has exactly two components. Since $X$ is good, it follows that $G'-X$ has exactly two components also. Hence (\ref{eq:two-seps}) holds. In particular, this implies that every good minimal separator of $G'$ is inclusion-minimal, and so (iv) is satisfied.

It remains only to prove Claim~\ref{claim:core-ms-lem}. We shall first prove that most minimal separators of $G'$ are minimal $(b,c)$-separators for some $b,c \in V$ (see Subclaim~\ref{subclaim:core-restrict-to-bc}). We shall then prove that most such minimal separators $X$ of $G'$ maximise $|\pi(X)|$ (see Subclaim~\ref{subclaim:core-most-project-large}). Finally, we shall prove that if $X$ does maximise $|\pi(X)|$ then $X$ is good (see Subclaim~\ref{subclaim:core-project-large-good}). The first part of Claim~\ref{claim:core-ms-lem} will therefore follow easily. Moreover, Subclaims~\ref{subclaim:core-most-project-large} and~\ref{subclaim:core-project-large-good} will imply the second part of Claim~\ref{claim:core-ms-lem} in a similar fashion.

\begin{subclaim}\label{subclaim:core-restrict-to-bc}
There are at most $2^5mk$ minimal separators in $G'$ which are not minimal $(b,c)$-separators for some $b, c \in V$.
\end{subclaim}

\textit{Proof of Subclaim~\ref{subclaim:core-restrict-to-bc}:} Let $X$ be a minimal $(b,c)$-separator in $G'$ for some $b,c \in V(G')$. We say $X$ is \emph{trivial} if $X \subseteq V(P^e_i)$ for some $e \in E$, $i \in [k]$. We claim that if $X$ is not a $(b',c')$-separator for some $b',c' \in V$ then $X$ is trivial, from which the result follows.

Suppose without loss of generality that $b$ is an internal vertex of $P^e_i$ for some $e \in E$, $i \in [k]$. Suppose that the component of $G'-X$ containing $b$ is a subset of $W^e$. Then $X \cap V(P^e_i)$ is already a $(b,c)$-separator, and so by minimality we have $X \subseteq V(P^e_i)$. Hence $X$ is trivial. We may therefore assume that the component of $G'-X$ containing $b$ also contains some endpoint $b' \in V$ of $e$.

If $c \in V$ then $X$ is a minimal $(b',c)$-separator and we are done. If $c$ is an internal vertex of $V(P^f_j)$ for some $f\in E$, $j \in [k]$, then by the same argument either $X$ is trivial or there exists $c' \in V$ such that $c'$ and $c$ lie in the same component of $G'-X$. Thus $X$ is either trivial or a minimal $(b',c')$-separator, as required. We have therefore proved Subclaim~\ref{subclaim:core-restrict-to-bc}.

\begin{subclaim}\label{subclaim:core-most-project-large}
Let $a \in \N$, and let $b,c \in V$ be distinct. There are at most $2^{m+n}3^{k(a-1)}$ minimal $(b,c)$-separators $X$ of $G'$ with $|\pi(X)| < a$.
\end{subclaim}

\textit{Proof of Subclaim~\ref{subclaim:core-most-project-large}:} We may choose any minimal $(b,c)$-separator $X$ of $G'$ by choosing first $X \cap V$, then $\pi(X)$, then $X \cap W^e$ for each $e \in \pi(X)$. There are at most $2^n$ ways of choosing $X \cap V$ and at most $2^m$ ways of choosing $\pi(X)$. For each $e \in \pi(X)$, since $b,c \in V$, $X$ must contain exactly one vertex internal to each $P^e_i$ and so there are exactly $3^k$ ways of choosing $X \cap W^e$. Since $|\pi(X)| \le a-1$, Subclaim~\ref{subclaim:core-most-project-large} follows.
 
\begin{subclaim}\label{subclaim:core-project-large-good}
Let $b,c \in V$ be distinct, and let $z$ be the maximum cardinality of any minimal $(b,c)$-edge separator of $G$. If $X$ is a minimal $(b,c)$-separator of $G'$ with $|\pi(X)| \ge z$, then $X$ is $z$-good.
\end{subclaim}

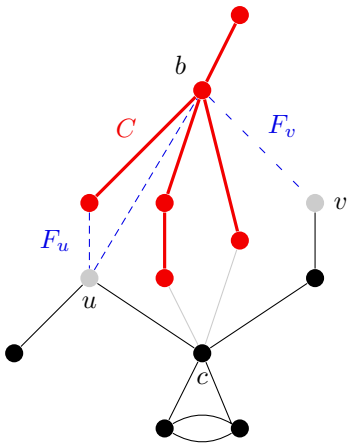
\begin{figure}
\begin{center}
\begin{tikzpicture}

\node (G11) at (2, 0)   [vertex] {};
\node (G12) at (3, 0)   [vertex] {};
\node (G21) at (0, 1)   [vertex] {};
\node (G22) at (2.5, 1) [vertex, label=below:$c$] {};
\node (G31) at (1, 2)   [vertex, color=medgrey, label=below:$u$] {};
\node (G32) at (2, 2)   [vertex, color=darkred] {};
\node (G33) at (4, 2)   [vertex] {};
\node (G41) at (3, 2.5) [vertex, color=darkred] {};
\node (G51) at (1, 3)   [vertex, color=darkred] {};
\node (G52) at (2, 3)   [vertex, color=darkred] {};
\node (G53) at (4, 3)   [vertex, color=medgrey, label=right:$v$] {};
\node (G61) at (2.5, 4.5) [vertex, color=darkred, label=above left:$b$] {};
\node (G71) at (3, 5.5)   [vertex, color=darkred] {};

\draw (G21) -- (G31) -- (G22) -- (G33) -- (G53);
\draw (G22) -- (G11) to [bend left] (G12) -- (G22);
\draw (G11) to [bend right] (G12);
\draw [color=darkred, very thick] (G71) -- (G61) -- (G51);
\draw [color=darkred, very thick] (G61) -- (G41);
\draw [color=darkred, very thick] (G61) -- (G52) -- (G32);
\draw [color=medgrey] (G41) -- (G22) -- (G32);
\draw [color=blue, densely dashed] (G51) -- (G31) -- (G61);
\draw [color=blue, loosely dashed] (G61) -- (G53);

\node at ($(G31)!0.5!(G51)$) [label={[blue,label distance=0]180:$F_u$}] {};
\node at ($(G61)!0.5!(G53)$) [label={[blue,label distance=-5]45: $F_v$}] {};
\node at ($(G51)!0.5!(G61)$) [label={[darkred, label distance=-5]135:$C$}] {};

\end{tikzpicture}
\end{center}

\caption{An example of the minimal $(b,c)$-edge separator $F$ of $G$ formed in the proof of Subclaim~\ref{subclaim:core-project-large-good}. The grey vertices and edges are elements of the hybrid minimal $(b,c)$-separator of $G$ corresponding to $X$. Thus $F$ consists of the grey edges of $G$ together with the edges in $F_u$ and $F_v$.}
\label{fig:vertex-sep}
\end{figure}

\textit{Proof of Subclaim~\ref{subclaim:core-project-large-good}:} Note that since $b,c \in V$, if $X \cap W^e \ne \emptyset$ for some $e \in E$ then $|X \cap \{x^e_i, y^e_i, z^e_i\}| = 1$ for all $i \in [k]$. In particular, if $X$ satisfies (c) then $X$ satisfies (a) and (b). To prove that $X$ satisfies (c) and (d), we shall exhibit a minimal $(b,c)$-edge separator $F$ of $G$ with cardinality at least $|\pi(X)| + |X \cap V|$. Thus (c) and (d) will follow from the definition of $z$ and the fact that $|\pi(X)| \ge z$. See Figure~\ref{fig:vertex-sep} for an example.

We say a pair $(Y, D)$ with $Y \subseteq V$ and $D \subseteq E$ is a \emph{hybrid minimal $(b,c)$-separator} of $G$ if it satisfies the following properties.
\begin{enumerate}
\item[(P1)] $b$ and $c$ lie in separate components of $(G - D) - Y$.
\item[(P2)] For all $Y' \subset Y$, $b$ and $c$ lie in the same component of $(G-D)-Y'$.
\item[(P3)] For all $D' \subset D$, $b$ and $c$ lie in the same component of $(G-D')-Y$.
\end{enumerate}
Thus $(X \cap V, \pi(X))$ is a hybrid minimal $(b,c)$-separator of $G$, since $X$ is a minimal $(b,c)$-separator of $G'$. 

Let $C$ be the component of $(G-\pi(X)) - (X \cap V)$ containing $b$. For each $v \in X \cap V$, let $F_v\subseteq E$ be the multiset of edges between $v$ and $C$ in $G$. Let $F = \pi(X) \cup \bigcup_{v \in X \cap V} F_v$. We claim that $F$ is the required minimal $(b,c)$-edge separator of $G$.

Note that $F_u \cap F_v = \emptyset$ for all distinct $u,v \in X \cap V$. For all $v \in  X \cap V$, we must have $F_v \ne \emptyset$ or (P2) would be violated on taking $Y' = (X \cap V) \setminus \{v\}$. Moreover, we must have $F_v \cap \pi(X) = \emptyset$ or (P3) would be violated on taking $D' = \pi(X) \setminus F_v$. Hence $|F| \ge |\pi(X)| + |X \cap V|$ as required. 

It is immediate from (P1) that $F$ is a $(b,c)$-edge separator of $G$. Finally, note that $F$ is minimal --- (P2) implies that no edge in any $F_v$ can be removed from $F$, and (P3) implies that no edge in $\pi(X)$ can be removed from $F$. Thus $F$ is a minimal $(b,c)$-edge separator of cardinality at least $|\pi(X)| + |X \cap V|$, and so Subclaim~\ref{subclaim:core-project-large-good} follows.

We now prove the first part of Claim~\ref{claim:core-ms-lem}. By Subclaim~\ref{subclaim:core-restrict-to-bc}, all but at most $2^5mk$ minimal separators of $G'$ are minimal $(b,c)$-separators for some $b,c \in V$. Moreover, by Subclaim~\ref{subclaim:core-most-project-large}, there are at most $n^2 \cdot 2^{m+n} 3^{k(x-1)}$ such separators $X$ with $|\pi(X)| < x$. Finally, by Subclaim~\ref{subclaim:core-project-large-good} and the definition of $x$, if $|\pi(X)| \ge x$ then $X$ is $x$-good. It follows that all but at most $2^5mk + n^2 2^{m+n} 3^{k(x-1)}$ minimal separators of $G'$ are $x$-good. We have
\[2^5mk \le 2^5k^2 \le 2^{k+5} \le 3^{\frac{2}{3}(k+5)} \le 3^{k-2} \le \frac{3^{kx}}{8}\]
and
\[n^2 2^{m+n} 3^{k(x-1)} \le 3^{k-16} 3^{k(x-1)} \le \frac{3^{kx}}{8},\]
so all but at most $3^{kx}/4$ minimal separators of $G'$ are $x$-good as required.

The second part of Claim~\ref{claim:core-ms-lem} follows more easily. By Subclaim~\ref{subclaim:core-most-project-large}, all but at most $2^{m+n} 3^{k(y-1)}$ minimal $(s,t)$-separators $X$ of $G'$ satisfy $|\pi(X)| \ge y$. It therefore follows from Subclaim~\ref{subclaim:core-project-large-good} that all but at most 
\[2^{m+n} 3^{k(y-1)} \le \frac{3^{ky}}{4}\]
minimal $(s,t)$-separators of $G'$ are $y$-good as required. Thus Claim~\ref{claim:core-ms-lem} follows, as does the result.
\end{proof}

We can now prove Theorem~\ref{thm:ms-vertex-hard}.

{\renewcommand{\thethm}{\ref{thm:ms-vertex-hard}}
\begin{thm}
The problems 
\stbms, \bms, \stbmstwo, \bmstwo\
and \bims are \numP-complete and are equivalent to \sat under AP-reduction.
\end{thm}}
\begin{proof}
All five problems are in \numP, and hence AP-reducible to \sat by~\cite{dggj-approx}. As in the proof of Theorem~\ref{thm:maxl-bis-hard}, Lemma~\ref{lem:minsep-core-vertex} implies that 
\begin{align*}
\largemes &\lAP \bms, \\
\largemes &\lAP \bims, \\
\largestmes &\lAP \stbmes.
\end{align*}
Moreover, since $\TMES(G)$ and $\TMES(G,s,t)$ are integers for all $G$, $s$ and $t$, Lemma~\ref{lem:minsep-core-vertex} also yields exact Turing reductions. Finally, note that in the proof of Lemma~\ref{lem:minsep-core-vertex}, all good minimal separators of $G'$ separate $G'$ into two components (see (\ref{eq:two-seps})). Analogues of Lemma~\ref{lem:minsep-core-vertex}(ii)--(iv) for \stbmstwo and \bmstwo therefore follow instantly. The result now follows by Lemma~\ref{lem:largeminsep-hard}.
\end{proof}

\section{Problems related to \maxlbis}
\label{sec:related}

\subsection{Hardness of \bds}

Recall that \maxlbis can   be viewed as counting the number of independent dominating sets in a bipartite graph --- a combination of \bis and \bds. We shall now prove that \bds is \sat-hard. We shall reduce from the following problem, which is well-known to be \sat-hard in the guise of \is (see Theorem~\ref{thm:is-hard}).

\begin{defn}
Let $G = (V,E)$ be a graph. A set $S \subseteq V$ is a \emph{vertex cover} of $G$ if $e \cap S \ne \emptyset$ for all $e \in E$. We write $\VC(G)$ for the number of vertex covers of $G$.
\end{defn}

\begin{prob}\label{prob:vc}
\myprob{\vc}{A graph $G$}{The number of vertex covers of $G$, which we denote by $\VC(G)$}	
\end{prob}

We can now prove Theorem~\ref{thm:bds-hard}.

{\renewcommand{\thethm}{\ref{thm:bds-hard}}
\begin{thm}
$\bds \eAP \sat$.
\end{thm}}

\begin{figure}[t]
\begin{center}
\begin{tikzpicture}[scale=0.8]

\node (Ga) at (2, 3.464) [vertex, label = above:$1$]       {};
\node (Gb) at (4,0)      [vertex, label = below right:$2$] {};
\node (Gc) at (0,0)      [vertex, label = below left:$3$]  {};
\draw (Ga) -- (Gb) -- (Gc) -- (Ga);

\node (G'a) at (10, 3.464) [vertex, label = above right:$1$] {};
\node (G'b) at (12,0)      [vertex, label = below right:$2$] {};
\node (G'c) at (8,0)       [vertex, label = below left:$3$]  {};
\node (G's) at (10, 4.5)   [vertex, label = below left:$s$, color=blue] {};

\def \x {8mm}
\def \y {2.5mm}
\node (G'ab1) at ($(G'a)!0.5! (G'b)$) [smallvertex] {};
\node (G'ab2) at ($(G'a)!0.5! (G'b)!\x!90:(G'b)$) [smallvertex] {};
\node (G'ab3) at ($(G'a)!0.5! (G'b)!\x+\y!90:(G'b)$) [smallvertex] {};

\node (G'bc1) at ($(G'b)!0.5! (G'c)$) [smallvertex] {};
\node (G'bc2) at ($(G'b)!0.5! (G'c)!\x!90:(G'c)$) [smallvertex] {};
\node (G'bc3) at ($(G'b)!0.5! (G'c)!\x+\y!90:(G'c)$) [smallvertex] {};

\node (G'ca1) at ($(G'c)!0.5! (G'a)$) [smallvertex] {};
\node (G'ca2) at ($(G'c)!0.5! (G'a)!\x!270:(G'c)$) [smallvertex] {};
\node (G'ca3) at ($(G'c)!0.5! (G'a)!\x+\y!270:(G'c)$) [smallvertex] {};

\node (G'y1)   at (9,    6) [smallvertex, color=blue] {};
\node (G'y2)   at (9.5,  6) [smallvertex, color=blue] {};
\node (G'yk-1) at (10.5, 6) [smallvertex, color=blue] {};
\node (G'yk)   at (11,   6) [smallvertex, color=blue] {};

\def \ew {1mm}
\node (G'abdot2) at ($(G'ab1)!0.5!(G'ab2)$)   [ellipses] {};
\node (G'abdot1) at ($(G'abdot2)!\ew!(G'ab1)$) [ellipses] {};
\node (G'abdot3) at ($(G'abdot2)!\ew!(G'ab2)$) [ellipses] {};

\node (G'bcdot2) at ($(G'bc1)!0.5!(G'bc2)$)   [ellipses] {};
\node (G'bcdot1) at ($(G'bcdot2)!\ew!(G'bc1)$) [ellipses] {};
\node (G'bcdot3) at ($(G'bcdot2)!\ew!(G'bc2)$) [ellipses] {};

\node (G'cadot2) at ($(G'ca1)!0.5!(G'ca2)$)   [ellipses] {};
\node (G'cadot1) at ($(G'cadot2)!\ew!(G'ca1)$) [ellipses] {};
\node (G'cadot3) at ($(G'cadot2)!\ew!(G'ca2)$) [ellipses] {};

\node (G'ydot2) at ($(G'y1)!0.5!(G'yk)$)   [ellipses] {};
\node (G'ydot1) at ($(G'ydot2)!\ew!(G'y1)$) [ellipses] {};
\node (G'ydot3) at ($(G'ydot2)!\ew!(G'yk)$) [ellipses] {};

\draw (G'a) -- (G'ab1) -- (G'b);
\draw (G'a) -- (G'ab2) -- (G'b);
\draw (G'a) -- (G'ab3) -- (G'b);

\draw (G'b) -- (G'bc1) -- (G'c);
\draw (G'b) -- (G'bc2) -- (G'c);
\draw (G'b) -- (G'bc3) -- (G'c);

\draw (G'c) -- (G'ca1) -- (G'a);
\draw (G'c) -- (G'ca2) -- (G'a);
\draw (G'c) -- (G'ca3) -- (G'a);

\draw [color=blue, dashed] (G'a) -- (G's);
\draw [color=blue, dashed] (G'b) to [out=45,  in=0]   (G's);
\draw [color=blue, dashed] (G'c) to [out=135, in=180] (G's);

\draw [color=blue, dashed] (G's) -- (G'y1);
\draw [color=blue, dashed] (G's) -- (G'y2);
\draw [color=blue, dashed] (G's) -- (G'yk-1);
\draw [color=blue, dashed] (G's) -- (G'yk);

\node at (Gb)  [occupied] {};
\node at (Gc)  [occupied] {};
\node at (G'b) [occupied] {};
\node at (G'c) [occupied] {};
\node at (G's) [occupied] {};

\node at (G'ab3) [smalloccupied] {};
\node at (G'bc1) [smalloccupied] {};
\node at (G'ca1) [smalloccupied] {};
\node at (G'ca2) [smalloccupied] {};
\node at (G'y1)  [smalloccupied] {};
\node at (G'yk-1)[smalloccupied] {};
\node at (G'yk)  [smalloccupied] {};

\node (Glabel)  at (2, 7)  [font=\Large] {$G$};
\node (G'label) at (10, 7) [font=\Large] {$G'$};

\end{tikzpicture}
\end{center}

\caption{An example of the reduction from an instance $G$ of \vc to an instance $G'$ of \bds used in the proof of Theorem~\ref{thm:bds-hard}. The boxes around vertices indicate a vertex cover in $G$ and a corresponding dominating set in $G'$. Note in particular how the presence of $s$ ensures that vertices 1, 2 and 3 are dominated in $G'$.}
\label{fig:bds}
\end{figure}
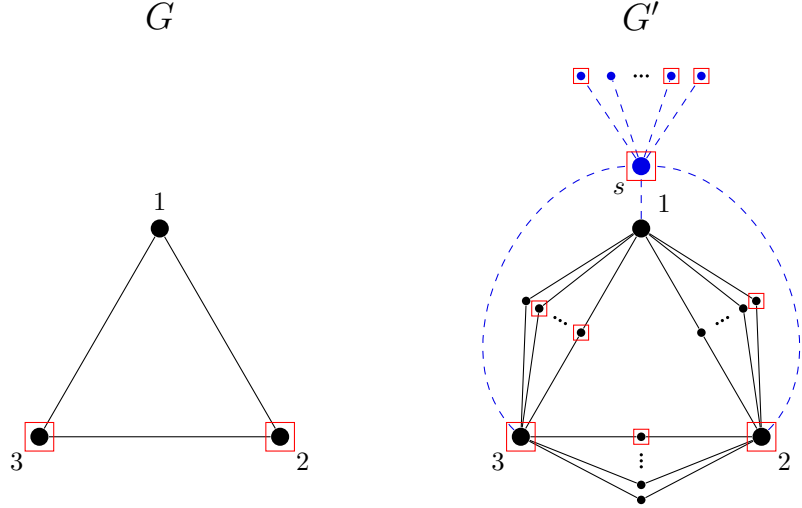

\begin{proof}
For every instance $G'$ of \bds, let $\DS(G')$ be the number of dominating sets in $G'$. Since \bds is in \numP, $\bds \lAP \sat$ follows from~\cite{dggj-approx}. We will show $\vc \lAP \bds$. Let $G = (V,E)$ be an instance of \vc. Without loss of generality, let $V = [n]$ for some $n \in \N$, let $m = |E|$, and let $t = \lceil n+\log_2(m+1)+3\rceil$. We shall construct an instance $G'$ of \bds with the property that $\VC(G) \le \DS(G')/2^{(m+1)t} \le \VC(G)+\frac{1}{4}$, which will be sufficient for the reduction as in the proof of Theorem~\ref{thm:maxl-bis-hard}. See Figure~\ref{fig:bds} for an example.

Informally, we obtain a bipartite graph $G'$ (an instance of \bds) from $G$ by first thickening and then 2-stretching each edge, then adding a gadget to $G$'s vertices which will allow us to ignore their domination constraints. Formally, we define $G'$ as follows. For each $e \in E$ let $X_e$ be a set of $t$ vertices, disjoint from $[n]$, where $X_e \cap X_f = \emptyset$ whenever $e \ne f$. Let $W = \bigcup_{e \in E}X_e$. Let $Y$ be a set of $t$ vertices disjoint from $[n] \cup W$, and let $s$ be a vertex not contained in $[n] \cup W \cup Y$. Then we define
\begin{align*}
V(G') & = Y \cup \{s\} \cup [n] \cup W, \\
E(G') & = \{\{y,s\}\mid  y \in Y\} \cup
		   \{\{s,i\}\mid  i \in [n]\} \cup 
		   \bigcup_{e = \{i,j\} \in E} \{\{i,x\},\{x,j\}\mid x \in X_e\}.
\end{align*}

We say a dominating set $S \subseteq V(G')$ is \emph{good} if the following conditions hold.
\begin{enumerate}
\item $s \in S$.
\item For all $e \in E$, we have $e \cap S \ne \emptyset$.
\end{enumerate}
We will show that good dominating sets in $G'$ correspond to vertex covers in $G$, and that almost all dominating sets in $G'$ are good.

First note that there are exactly $2^{(m+1)t}$ ways of extending any vertex cover $X$ of $G$ into a good dominating set of $G'$. Indeed, a set $S$ satisfying $X \cap [n] = S$ is a good dominating set of $G'$ if and only if $s \in S$. Hence there are $2^{(m+1)t}\VC(G)$ good dominating sets of $G'$, and in particular $\DS(G')/2^{(m+1)t} \ge \VC(G)$.

Moreover, suppose that $S$ is a dominating set of $G'$ which is not good. Then either $s \notin S$ or $e \cap S = \emptyset$ for some $e \in E$. If $s \notin S$, then $Y \subseteq S$. If $e \cap S = \emptyset$ for some $e \in E$, then $X_e \subseteq S$. Note that $n + \log_2(m+1) + 2 \le t$, so $2^{mt+n+\log_2(m+1)} \le 2^{(m+1)t-2}$. Hence there are at most
\[(m+1)2^{|V(G')| - t} = (m+1)2^{mt + n + 1} \le \frac{2^{(m+1)t}}{4}\]
dominating sets of $G'$ which are not good. In particular, we have
\[\frac{\DS(G')}{2^{(m+1)t}} \le \VC(G) + \frac{1}{4}.\]
The result therefore follows.
\end{proof}

\subsection{Hardness of \setunion}

We shall now prove that \setunion is \sat-hard by a reduction from \maxlbis. 

{\renewcommand{\thethm}{\ref{thm:su-hard}}
\begin{thm}
$\setunion \eAP \sat$.
\end{thm}}

\begin{figure}[t]
\begin{center}
\begin{tikzpicture}
\node (G1) at (0,3) [vertex, label = left:1] {};
\node (G2) at (0,2) [vertex, label = left:2] {};
\node (G3) at (0,1) [vertex, label = left:3] {};
\node (G4) at (0,0) [vertex, label = left:4] {};

\node (Ga) at (2,3)   [vertex, label = right:$a$] {};
\node (Gb) at (2,1.5) [vertex, label = right:$b$] {};
\node (Gc) at (2,0)   [vertex, label = right:$c$] {};

\draw (G2) -- (Ga) -- (G1) -- (Gb) -- (G3);
\draw (Gc) -- (G4) -- (Gb);

\draw (-0.2,1.5) ellipse (0.75 and 2.1);
\draw (2.2, 1.5) ellipse (0.75 and 2.1);

\node at (G2) [occupied] {};
\node at (Gb) [occupied] {};
\node at (Gc) [occupied] {};

\matrix[matrix of nodes, inner sep=0pt, column sep=0pt] at (6, 1.5) {
      \node (Fstart) {$\{$}; & 
      \node (Fa)   {$\{1,2\}$}; & 
      \node (Fab)  {$,\,$}; & 
      \node (Fb)   [draw, inner xsep = 1, inner ysep = 2, color=red, 
                         text=black] {$\{1,3,4\}$}; &
      \node (Fbc)  {$\,,$}; &
      \node (Fc)   [draw, inner xsep = 1, inner ysep = 2, color=red, 
                         text=black] {$\{4\}$}; &
      \node (Fend) {$\}$}; \\
};

\node at (Fa) [label={[label distance = 2]90:$a$}] {};
\node at (Fb) [label={[label distance = 2]90:$b$}] {};
\node at (Fc) [label={[label distance = 2]90:$c$}] {};

\node at (-0.2, -1) {$A$};
\node at (2.2, -1)  {$B$};
\node at (1, 4)  [font=\Large] {$G$};
\node at (6, 4)  [font=\Large] {$\F$};
\end{tikzpicture}
\end{center}

\caption{An example of the reduction from an instance $G$ of \maxlbis to an instance $(n,\F)$ of \setunion used in the proof of Theorem~\ref{thm:su-hard}. The boxes around vertices indicate a maximal independent set in $G$ and the corresponding union of sets in $\F$.}
\label{fig:setunion}
\end{figure}
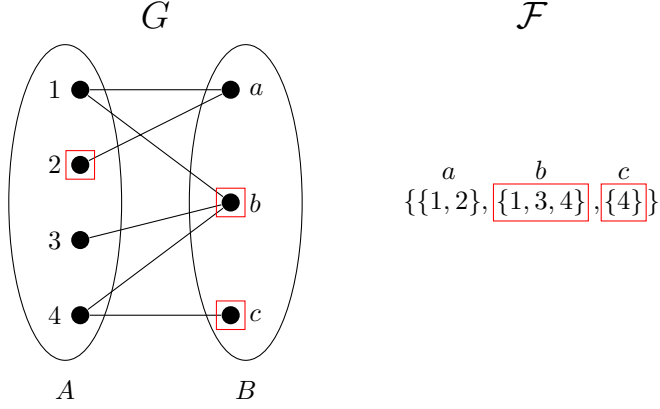

\begin{proof}
Since \setunion is in \numP, $\setunion \lAP \sat$ follows from~\cite{dggj-approx}. We will show $\maxlbis \lAP \setunion$. Let $G = (V,E)$ be an instance of \maxlbis with vertex classes $A$ and $B$. Note that \maxlbis is hard by Theorem~\ref{thm:maxl-bis-hard}. Without loss of generality, let $A = [n]$ for some $n \in \N$. We shall construct an instance $(n, \F)$ of \setunion with the property that $\MIS(G) = \SU{\F}$, from which the result follows immediately. See Figure~\ref{fig:setunion} for an example.

Let $\F = \{N(v)\mid v \in B\}$, so that 
\begin{equation}\label{eq:setunion-F}
\U(\F) = \{N(S)\mid  S \subseteq B\}.
\end{equation}
Given $S \subseteq [n]$, write $\overline{S} = [n] \setminus S$. Similarly, given $S \subseteq B$, write $\overline{S} = B \setminus S$. Given $S \subseteq [n]$, write $\MIS_S(G)$ for the number of maximal independent sets $X \subseteq V$ with $X \cap [n] = S$. We shall prove that
\begin{equation}\label{eq:setunion-target}
\MIS_S(G) = \begin{cases}1 & \textnormal{if }\overline{S} \in \U(\F), \\0 & \textnormal{otherwise.}\end{cases}
\end{equation}
It will follow immediately that the map $X \mapsto \overline{X \cap [n]}$ is a bijection from the set of maximal independent sets of $G$ to $\U(\F)$. 

Take $S \subseteq [n]$. Note that $\MIS_S(G) \in \{0,1\}$ --- any maximal independent set $X \subseteq V$ is uniquely determined by its intersection with $[n]$. It therefore suffices to prove that $\MIS_S(G) > 0$ if and only if $\overline{S} \in \U(\F)$. 

First suppose $\overline{S} \in \U(\F)$. Let $T\subseteq B$ be a maximal set such that $\overline{S} = N(T)$. Then it is immediate that there are no edges between $S$ and $T$. Moreover, $\overline{T} \subseteq N(S)$ by maximality of $T$. Hence $S \cup T$ is a maximal independent set in $G$, and $\MIS_S(G) = 1$.

Now suppose $\overline{S} \notin \U(\F)$, and suppose $X \subseteq V(G)$ is a maximal independent set of $G$ with $X \cap [n] = S$. Then by independence we have $N(S) \cap X = \emptyset$, and by maximality we have $\overline{N(S)} \subseteq X$. Thus $X = S \cup \overline{N(S)}$. By maximality, it follows that $\overline{S} \subseteq N(\overline{N(S)})$ --- otherwise an element of $\overline{S}$ could be added to $X$. But $N(\overline{N(S)}) \cap S = \emptyset$ since $\overline{N(S)}$ is precisely the set of vertices in $B$ with no edges to $S$, so $N(\overline{N(S)}) \subseteq \overline{S}$ and hence $\overline{S} = N(\overline{N(S)})$. But this implies $\overline{S} \in \U(\F)$ by equation~\eqref{eq:setunion-F}, which is a contradiction. Hence $\MIS_S(G) = 0$, and we have proved equation~\eqref{eq:setunion-target}. It follows that $\MIS(G) = \SU{\F}$, as required.
\end{proof}

\subsection{Hardness of \unionrep}

We shall now prove that \unionrep is \sat-hard by reducing from \vc.

{\renewcommand{\thethm}{\ref{thm:ur-hard}}
\begin{thm}
$\unionrep \eAP \sat$.
\end{thm}}

\begin{proof}
Since \unionrep is in \numP, $\unionrep \lAP \sat$ follows from~\cite{dggj-approx}. We will show $\vc \lAP \unionrep$. Let $G=(V,E)$ be an instance of \vc, which is hard by Theorem~\ref{thm:is-hard}. Without loss of generality let $V = [n]$ for some $n \in \N$,  and let $m = |E|$. We shall construct an instance $(m,\F)$ of \unionrep with the property that $\UR{[m]} = \VC(G)$, from which the result follows immediately.

For each $i \in [n]$, let $S_i$ be the set of edges incident to $i$ in $G$. Let $\F = \{S_i \mid i \in [n]\}$. Thus on identifying $E$ with $[m]$, $(m,\F)$ becomes an instance of \unionrep. Given a set $X \subseteq [n]$, let $X' = \{S_i: i \in X\}$. Then $X$ is a vertex cover of $G$ if and only if $\cup X' = E$. Thus there is a bijection between $\U_{\F}^{-1}(E)$ and vertex covers of $G$, and so $\UR{E} = \VC(G)$ as required.
\end{proof}

\section*{Index of problems}

\noindent\begin{tabular}{@{}p{\textwidth-\widthof{00}}@{}p{\widthof{00}}@{}}
\bmstwo (Problem \ref{prob:bmstwo}) \dotfill & \pageref{prob:bmstwo} \\
\stbmstwo (Problem \ref{prob:stbmstwo}) \dotfill & \pageref{prob:stbmstwo} \\
\bds (Problem \ref{prob:bds}) \dotfill & \pageref{prob:bds} \\
\bims (Problem \ref{prob:bims}) \dotfill & \pageref{prob:bims} \\
\bmes (Problem \ref{prob:bmes}) \dotfill & \pageref{prob:bmes} \\
\stbmes (Problem \ref{prob:stbmes}) \dotfill & \pageref{prob:stbmes} \\
\bms (Problem \ref{prob:bms}) \dotfill & \pageref{prob:bms} \\
\stbms (Problem \ref{prob:stbms}) \dotfill & \pageref{prob:stbms} \\
\bis (Problem \ref{prob:bis}) \dotfill & \pageref{prob:bis} \\
\is (Problem \ref{prob:is}) \dotfill & \pageref{prob:is} \\
\largemes (Problem \ref{prob:largemes}) \dotfill & \pageref{prob:largemes} \\
\largestmes (Problem \ref{prob:largestmes}) \dotfill & \pageref{prob:largestmes} \\
\maxlbis (Problem \ref{prob:maxlbis}) \dotfill & \pageref{prob:maxlbis} \\
\tnaesat (Problem \ref{prob:tnaesat}) \dotfill & \pageref{prob:tnaesat} \\
\setunion (Problem \ref{prob:setunion}) \dotfill & \pageref{prob:setunion} \\
\unionrep (Problem \ref{prob:unionrep}) \dotfill & \pageref{prob:unionrep} \\
\vc (Problem \ref{prob:vc}) \dotfill & \pageref{prob:vc} \\
\end{tabular}

\section*{Acknowledgements}

We thank Luca Manzoni and Yuri Pirola for useful discussions.

\bibliographystyle{plain}
\bibliography{\jobname}

\begin{thebibliography}{10}

\bibitem{BBC00}
Anne Berry, Jean~Paul Bordat, and Olivier Cogis.
\newblock Generating all the minimal separators of a graph.
\newblock {\em International Journal of Foundations of Computer Science},
  11(3):397--403, 2000.

\bibitem{BF05}
Hans~L. Bodlaender and Fedor~V. Fomin.
\newblock Tree decompositions with small cost.
\newblock {\em Discrete Applied Mathematics}, 145(2):143 -- 154, 2005.

\bibitem{BFKKT12}
Hans~L. Bodlaender, Fedor~V. Fomin, Arie~M.C.A. Koster, Dieter Kratsch, and
  Dimitrios~M. Thilikos.
\newblock On exact algorithms for treewidth.
\newblock {\em ACM Trans. Algorithms}, 9(1):12:1--12:23, 2012.

\bibitem{BK08}
Hans~L. Bodlaender and Arie~M.C.A. Koster.
\newblock Combinatorial optimization on graphs of bounded treewidth.
\newblock {\em The Computer Journal}, 51(3):255--269, 2008.

\bibitem{BT01}
Vincent Bouchitt{\'e} and Ioan Todinca.
\newblock Treewidth and minimum fill-in: grouping the minimal separators.
\newblock {\em SIAM Journal on Computing}, 31(1):212--232, 2001.

\bibitem{BT02}
Vincent Bouchitt{\'e} and Ioan Todinca.
\newblock Listing all potential maximal cliques of a graph.
\newblock {\em Theoretical Computer Science}, 276(1-2):17--32, 2002.

\bibitem{bcst-setunion}
Henning Bruhn, Pierre Charbit, Oliver Schaudt, and Jan~Arne Telle.
\newblock The graph formulation of the union-closed sets conjecture.
\newblock {\em European J. Combin.}, 43:210--219, 2015.

\bibitem{crescenzi-badap}
Pierluigi Crescenzi.
\newblock A short guide to approximation preserving reductions.
\newblock In {\em Twelfth {A}nnual {IEEE} {C}onference on {C}omputational
  {C}omplexity ({U}lm, 1997)}, pages 262--273. IEEE Computer Soc., Los
  Alamitos, CA, 1997.

\bibitem{diestel}
Reinhard Diestel.
\newblock {\em Graph Theory, 4th Edition}, volume 173 of {\em Graduate texts in
  mathematics}.
\newblock Springer, 2012.

\bibitem{dggj-approx}
Martin Dyer, Leslie~Ann Goldberg, Catherine Greenhill, and Mark Jerrum.
\newblock The relative complexity of approximate counting problems.
\newblock {\em Algorithmica}, 38(3):471--500, 2004.
\newblock Approximation algorithms.

\bibitem{FKTV08}
Fedor~V. Fomin, Dieter Kratsch, Ioan Todinca, and Yngve Villanger.
\newblock Exact algorithms for treewidth and minimum fill-in.
\newblock {\em SIAM Journal on Computing}, 38(3):1058--1079, 2008.

\bibitem{FV10}
Fedor~V. Fomin and Yngve Villanger.
\newblock Finding induced subgraphs via minimal triangulations.
\newblock In {\em 27th International Symposium on Theoretical Aspects of
  Computer Science, {STACS} 2010, March 4-6, 2010, Nancy, France}, pages
  383--394, 2010.

\bibitem{FV12}
Fedor~V. Fomin and Yngve Villanger.
\newblock Treewidth computation and extremal combinatorics.
\newblock {\em Combinatorica}, 32(3):289--308, 2012.

\bibitem{FY14}
Masanobu Furuse and Koichi Yamazaki.
\newblock A revisit of the scheme for computing treewidth and minimum fill-in.
\newblock {\em Theoretical Computer Science}, 531(0):66 -- 76, 2014.

\bibitem{leslielisting}
Leslie~Ann Goldberg.
\newblock {\em Efficient Algorithms for Listing Combinatorial Structures}.
\newblock Cambridge University Press, 1993.
\newblock Cambridge Books Online.

\bibitem{ICALP}
Leslie~Ann Goldberg, Rob Gysel, and John Lapinskas.
\newblock Approximately counting locally-optimal structures.
\newblock In {\em Automata, Languages, and Programming: 42nd International
  Colloquium, ICALP 2015, Kyoto, Japan, July 6-10, 2015, Proceedings, Part I},
  pages 654--665, Berlin, Heidelberg, 2015. Springer Berlin Heidelberg.

\bibitem{gj-ferroising}
Leslie~Ann Goldberg and Mark Jerrum.
\newblock The complexity of ferromagnetic {I}sing with local fields.
\newblock {\em Combin. Probab. Comput.}, 16(1):43--61, 2007.

\bibitem{gj-bisconjecture}
Leslie~Ann Goldberg and Mark Jerrum.
\newblock Approximating the partition function of the ferromagnetic {P}otts
  model.
\newblock {\em J. ACM}, 59(5):Art. 25, 31, 2012.

\bibitem{gj-notation}
Leslie~Ann Goldberg and Mark Jerrum.
\newblock The complexity of approximately counting tree homomorphisms.
\newblock {\em ACM Trans. Comput. Theory}, 2(2), 2014.

\bibitem{RobPreprint}
Rob Gysel.
\newblock Unique perfect phylogeny characterizations via uniquely representable
  chordal graphs.
\newblock {\em CoRR}, abs/1305.1375, 2013.

\bibitem{G14}
Rob Gysel.
\newblock Minimal triangulation algorithms for perfect phylogeny problems.
\newblock In {\em Proceedings of the 8th International Conference on Language
  and Automata Theory and Applications, Lecture Notes in Computer Science},
  volume 8370, pages 421--432. Springer, 2014.

\bibitem{JYPPLS}
David~S. Johnson, Christos~H. Papadimitriou, and Mihalis Yannakakis.
\newblock How easy is local search?
\newblock {\em J. Comput. Syst. Sci.}, 37(1):79--100, 1988.

\bibitem{kmnf}
Toshinobu Kashiwabara, Sumio Masuda, Kazuo Nakajima, and Toshio Fujisawa.
\newblock Generation of maximum independent sets of a bipartite graph and
  maximum cliques of a circular-arc graph.
\newblock {\em J. Algorithms}, 13(1):161--174, 1992.

\bibitem{kou}
Shuji Kijima, Yoshio Okamoto, and Takeaki Uno.
\newblock Dominating set counting in graph classes.
\newblock In {\em Computing and Combinatorics - 17th Annual International
  Conference, {COCOON} 2011, Dallas, TX, USA, August 14-16, 2011. Proceedings},
  pages 13--24, 2011.

\bibitem{KK98}
Ton Kloks and Dieter Kratsch.
\newblock Listing all minimal separators of a graph.
\newblock {\em SIAM Journal on Computing}, 27:605--613, 1998.

\bibitem{L10}
Daniel Lokshtanov.
\newblock On the complexity of computing treelength.
\newblock {\em Discrete Applied Mathematics}, 158(7):820 -- 827, 2010.

\bibitem{sy-pls}
Alejandro~A. Sch{\"a}ffer and Mihalis Yannakakis.
\newblock Simple local search problems that are hard to solve.
\newblock {\em SIAM J. Comput.}, 20(1):56--87, 1991.

\bibitem{SL97}
Hong Shen and Weifa Liang.
\newblock Efficient enumeration of all minimal separators in a graph.
\newblock {\em Theoretical Computer Science}, 180(1-2):169--180, 1997.

\bibitem{sipser}
Michael Sipser.
\newblock {\em Introduction to the Theory of Computation}.
\newblock Cengage Learning, 2nd edition, 2005.

\bibitem{Takata-seps}
Ken Takata.
\newblock Space-optimal, backtracking algorithms to list the minimal vertex
  separators of a graph.
\newblock {\em Discrete Appl. Math.}, 158(15):1660--1667, 2010.

\bibitem{tsukiyama}
Shuji Tsukiyama, Mikio Ide, Hiromu Ariyoshi, and Isao Shirakawa.
\newblock A new algorithm for generating all the maximal independent sets.
\newblock {\em {SIAM} J. Comput.}, 6(3):505--517, 1977.

\bibitem{vadhan}
Salil~P. Vadhan.
\newblock The complexity of counting in sparse, regular, and planar graphs.
\newblock {\em {SIAM} J. Comput.}, 31(2):398--427, 2001.

\bibitem{whitney-seps}
Hassler Whitney.
\newblock Planar graphs.
\newblock {\em Fundamenta Mathematicae}, 21(1):73--84, 1933.

\bibitem{zuckerman-sat}
David Zuckerman.
\newblock On unapproximable versions of {NP}-complete problems.
\newblock {\em SIAM J. Comput.}, 25(6):1293--1304, 1996.

\end{thebibliography}
 \end{document}